\documentclass[aps,prx,twocolumn]{revtex4-2}
\usepackage{amsmath,amsfonts,amssymb,amsthm,bm}
\usepackage{bbm,nccmath,accents}
\usepackage{stmaryrd,mathrsfs,mathtools}
\usepackage{tikz}
\usetikzlibrary{cd,calc}
\usepackage{array} 
\usepackage{multirow}
\usepackage{wasysym}
\usepackage{enumitem}
\usepackage[most]{tcolorbox}
\usepackage{hyperref}
\usepackage{color}
\usepackage{titlesec}
\usepackage[leqno,fleqn,intlimits]{empheq}  
\usepackage{atbegshi}
\usepackage{babel}

\newcommand*{\widefbox}[1]{ \fbox{\hspace{2em}#1\hspace{2em}} }

\tikzset{
  symbol/.style={
    draw=none,
    every to/.append style={
      edge node={node [sloped, allow upside down, auto=false]{$#1$}}}
  }
}

\def\presuper#1#2%
  {\mathop{}%
   \mathopen{\vphantom{#2}}^{#1}%
   \kern-\scriptspace%
   #2}
   
\newcommand{\paren}[1]{{\normalfont{(}}{#1}{\normalfont{)}}}

\newcommand{\defeq}{ {\kern 0.2em}:{\kern -0.5em}={\kern 0.2 em} }  
\newcommand{\eqdef}{ {\kern 0.2em}={\kern -0.5em}:{\kern 0.2 em} }  
\newcommand{\setof}[2]{\left\{ #1 \,\colon\, #2 \right\}}  

\newcommand{\tpose}{\prescript{t}{}}

\DeclareMathOperator{\rk}{rk}  
\DeclareMathOperator{\Id}{\mathbbm{1}} 

\newcommand{\Real}{{\mathbb R}}    
\newcommand{\Cmplx}{{\mathbb C}}   

\newcommand{\inpr}[2]{\langle {#1} \,|\, {#2} \rangle}

\newcommand{\Dket}[1]{\left| {#1} \right\rangle}                
\newcommand{\Dbra}[1]{\left\langle {#1} \right|}                
\DeclareMathOperator{\Tr}{\mathrm{Tr}}

\newcommand{\ee}{\mathsf{e}}
\newcommand{\kk}{\mathsf{k}}
\newcommand{\hh}{\mathsf{h}}
\newcommand{\mm}{\mathsf{m}}
\newcommand{\nn}{\mathsf{n}}

\newcommand{\cE}{\mathcal{E}}
\newcommand{\cH}{\mathcal{H}}
\newcommand{\cK}{\mathcal{K}}

\newcommand{\cV}{\mathcal{V}}

\newcommand{\cQ}{\mathcal{Q}}
\newcommand{\cM}{\mathcal{M}}
\newcommand{\cN}{\mathcal{N}}

\newcommand{\gC}{\mathfrak{C}}
\newcommand{\gJ}{\mathfrak{J}}

\newcommand{\gS}{\mathfrak{S}}

\newcommand{\gT}{\mathfrak{T}}

\newcommand{\Blin}%
{\mathsf{L}}

\newcommand{\CP}{\ensuremath{\mathsf{CP}}}

\newcommand{\sJ}{\mathscr{J}}

\newcommand{\jam}{\gJ}
\newcommand{\pam}{\overline}

\newcommand{\bijam}{ \;\xleftrightarrow{\;\jam\;}\; } 

\newcommand{\blank}{\,\square\,}

\newcommand{\bpi}{\bm{\pi}}
\newcommand{\btheta}{{\bm\theta}}
\DeclareMathOperator{\extr}{\mathrm{extr}}

\newcommand{\vJ}{{\mathcal J}}
\newcommand{\onbsum}[1]%
           { \sum_{{#1}}^{\scriptscriptstyle{\text{ONB}}} }

\theoremstyle{plain}
\newtheorem{thm}{Theorem}

\newtheorem{lem}[thm]{Lemma}
\newtheorem{prop}[thm]{Proposition}

\newtheorem*{lem*}{Lemma}
\newtheorem*{cor*}{Corollary}
\newtheorem*{thm*}{Theorem}
\newtheorem*{prop*}{Proposition}
\theoremstyle{definition}
\newtheorem{bareDefinition}{Definition}[section]
\newtheorem*{bareDefinition*}{Definition}

\newenvironment{definition*}
    {\begin{bareDefinition*}}
    {\hfill$\dashv$ \end{bareDefinition*}}

\newtheorem{bareExample}{Example}[section]

\theoremstyle{remark}

\newtheorem*{remark*}{Remark}

\let\origmaketitle\maketitle
\def\maketitle{
  \begingroup
  \def\uppercasenonmath##1{} 
  \let\MakeUppercase\relax 
  \origmaketitle
  \endgroup
}
\makeatletter
\newcommand*{\textoverline}[1]{$\overline{\hbox{#1}}\m@th$}
\makeatother
\newcolumntype{L}{>{$}l<{$}}
\newcolumntype{R}{>{$}r<{$}}
\newcolumntype{C}{>{$}c<{$}}
\newtcolorbox{stdbox}{
center,
width=1.0\linewidth,
halign=center,
colframe=black!40!white,
colback=white,
toprule=0.8pt,
bottomrule=0.8pt,
leftrule=0.8pt,
rightrule=0.8pt
}

\newsavebox\boxpi
\sbox\boxpi{\ensuremath{\bpi}}
\newcommand{\pibox}{\usebox{\boxpi}}
\newsavebox\boxtheta
\sbox\boxtheta{\ensuremath{\btheta}}
\newcommand{\thetabox}{\usebox{\boxtheta}}

\begin{document}

\title{Choi-Jamio{\l}kowski isomorphism(s) and channel/state duality,
naturally!}
  \author{Paul~E.~Lammert}
 \email{lammert@psu.edu}
\affiliation{
  Department of Physics \\ 
Pennsylvania State University \\ University Park, PA 16802-6300}
\date{\today}
\begin{abstract}

Provided a maximally entangled state vector in $\cH\otimes\cH$,
the {Choi-Jamio{l}kowski isomorphism} of quantum information theory
(or {channel/state duality}) produces a mapping of $\CP(\cH,\cK)$
(completely positive maps of density operators on $\cH$ to
those on $\cK$) onto the density operators on $\cK\otimes\cH$,
which extends to a unitary map. Here, $\cH$ and $\cK$ are
\textit{finite-dimensional} Hilbert spaces.
The nature of this construction, the form of the associated formulas,
and the possibility of versions which do not require an input
entangled vector, are illuminated from the perspective of \textit{naturality},
in the sense of Eilenberg and MacLane.
There is a unique such correspondence (``Choi map'') natural in both
$\cH$ and $\cK$ which produces density operators on either $\Blin(\cH,\cK)$
(Hilbert-Schmidt inner product) or $\cK\otimes\cH^*$ instead of
$\cK\otimes\cH$. In particular, no entangled state input is required.
There are no such \textit{binatural} maps carried on $\cK\otimes\cH$,
but those natural in only $\cK$ are canonically associated with
unitary maps $\cH\rightarrow\cH^*$ and derive from the binatural maps.
The construction recovers the standard Choi-Jamio{l}kowski formulas
and shows how they express naturality.
\end{abstract}
\maketitle

\tableofcontents
\section{Introduction}

Completely positive ($\CP$) maps are
a cornerstone of the theory of open quantum systems because they
constitute the unique class of maps on density operators which behave 
acceptably vis-\`a-vis composite systems.
However, getting a solid grip on them and their properties directly
is difficult.
Choi-Jamio{\l}kowski isomorphism and related notions of Choi matrix and
channel/state duality\cite{%
    de-Pillis-67,
    Jamiolkowski-72,
    Choi-75,
    Zyczkowski+Bengtsson-04,
    Arrighi+Patricot-04,
    Bengtsson+Zyczkowski-Book,
    Ranade+Ali-07,
    Grabowski+Kus+Marmo-07,
  Jiang+Luo+Fu-13,
  Majewski+Tylec-13b,
  Preskill-course,
  Wilde,
  D'Ariano+Chiribella+Perinotti-17,
  Milz+Pollock+Modi-17,
  Benatti-23,
  Vacchini-Foundations,
  Frembs+Cavalcanti-24,
Homa+Ortega+Koniorczyk-24%
}
provide valuable help with this problem.
Here is a typical statement.
  \begin{thm*}[CJ]
    With $\bpi(\Upsilon)$ the projection on a maximally entangled state
    $\Upsilon$ in \mbox{$\cH\otimes\cH$},
\begin{equation}\label{eq:usual Choi}
\begin{split}
&      \vJ_\Upsilon\colon \Blin(\Blin(\cH),\Blin(\cK)) \rightarrow 
      \Blin(\cK\otimes\cH) \\
&      \vJ_\Upsilon\Lambda = \Lambda\otimes\Id_{\Blin(\cH)}[\bpi(\Upsilon)]
\end{split}
\end{equation}
is an isomorphism.
Its inverse is given by,
    \begin{equation}\label{eq:usual Choi inverse}
     (\vJ_\Upsilon^{-1}\Psi)[Y] = n\Tr_{\cH}
        \left\{(\Id_{\cK}\otimes Y^T)\Psi\right\}
    \end{equation}
for \mbox{$Y\in\Blin(\cH)$}.

Furthermore, the map $\Lambda$ is CP if and only if the Choi matrix
    $\vJ_\Upsilon\Lambda$ is positive, and.
  \end{thm*}
  Evidently, $\vJ_\Upsilon$ depends on some arbitrary choice of an
  entangled state. In (\ref{eq:usual Choi inverse}), this dependence
  is hidden in the transpose operation.
  Are there transformations which achieve the \textit{same thing}, but 
  invariantly, or \textit{naturally}? Or, even without that virtue?
  Last, but not least, where do the formulas
  (\ref{eq:usual Choi},\ref{eq:usual Choi inverse}) really come from?
  These questions are related, and the crux of the matter
  is an explanation of what ``same thing'' and ``natural'' mean. 
To get started, consider this proposal for ``same thing''.
\begin{definition*}
An $(\cH,\cK)$-\textbf{Choi map} 
is a \textit{unitary} mapping
\begin{equation}\label{eq:Choi map}
\Blin(\Blin(\cH),\Blin(\cK)) \rightarrow \Blin(\cM),
\end{equation}
which maps the completely positive maps $\CP(\cH,\cK)$ into
the positive operators $\Blin(\cM)^+$ on some Hilbert space $\cM$,
called the \textbf{carrier} of the Choi map.
\end{definition*}
(Here, $\Blin(\Blin(\cH),\Blin(\cK))$
is a Hilbert space via the Hilbert-Schmidt inner product,
as reviewed in Section \ref{sec:Hilbert}).
The carrier of the conventional Choi-Jamio{l}kowski isomorphism
occuring in the CJ Theorem is $\cK\otimes\cH$.

But surely, the reader may be wont to object --- and rightly so --- this is an
excessive broadening. Should we not demand the carrier be \textit{naturally}
(that word again) related to $\cH$ and $\cK$?
No need for worry. The following proposition explains in what sense such
permissiveness does not allow anything essentially new, anyway.
\begin{prop}\label{prop:Choi}
  If $S$ and $T$ are $(\cH,\cK)$-Choi maps with carriers
  $\cM$ and $\cN$, respectively, then
  there is a unitary $U\colon\cN\rightarrow\cM$
  such that
  for $\Lambda$ in $\Blin(\Blin(\cH),\Blin(\cK))$,
  $S\Lambda = U(T\Lambda)U^\dag$.
\end{prop}
So, all $(\cH,\cK)$-Choi maps are essentially equivalent.
Yet, not all are equally natural. Are any natural?
In fact, there are two which are both simple,
and \textit{natural} in the sense famously expounded by Eilenberg
and MacLane\cite{Eilenberg+MacLane-45}.
This notion of naturality is close kin to the principle (Curie's, more or less)
that the symmetries of relevant laws must be reflected in cause-effect
linkages.
In the current context, transformation under mappings between,
possibly distinct, underlying Hilbert spaces are important.
The two mentioned Choi maps, $\gJ$ and $\gC$, have carrier
spaces $\Blin(\cH,\cK)$ and $\cK\otimes\cH^*$.
In fact, they are the unique Choi maps having
those carriers which are natural in both $\cH$ and $\cK$.
Formulas resembling (\ref{eq:usual Choi}) and (\ref{eq:usual Choi inverse})
hold for $\gC$; they are straightforwardly expressions of naturality.
Furthermore, $\gC$ can be simply modified with the aid of
a unitary $V\colon\cH\rightarrow\cH^*$ to obtain a Choi map 
$\gC_V$ with carrier $\cK\otimes\cH$.
It is essentially $\vJ_\Upsilon$
with $\Upsilon$ the entangled state determined by $V$.
$\gC_V$ cannot be natural in $\cH$, because that is locked down,
but it is natural in $\cK$. Moreover, the $\gC_V$, as $V$ varies over
unitaries $\cH\rightarrow\cH^*$, exhaust the natural-in-$\cK$ Choi maps
with carrier $\cK\otimes\cH$.

Conventional CJ isomorphisms 
achieve strict channel/state duality. That is, applied to a $\CP$ map
which is also trace-preserving, the result is a properly normalized
density matrix. Since this is an important feature for physics, the
reader may be worried whether our more capacious notion of Choi map
shares it. Fear not: the conclusion of Proposition~\ref{prop:Choi} shows
that all $(\cH,\cK)$-Choi maps are equally good (or bad) in this respect.

\subsubsection*{Brief guide to the contents}

Here is a brief overview of the paper body.
Section \ref{sec:Hilbert} reviews the Hilbert space structures on
dual spaces and $\Blin(\cH,\cK)$. Use of the Hilbert-Schmidt inner product
here means that our considerations are restricted to finite dimension.
Section \ref{sec:cats} explains the notion of naturality,
without assuming any knowledge of category theory but exposing
just what is needed.
Section \ref{sec:T} demonstrates our methods on the natural unitary
\mbox{$\gT\colon\Blin(\cH,\cK)\rightarrow\cK\otimes\cH^*$},
which is an important building block. $\gJ$ is introduced in
\S\ref{sec:J} and shown to be a Choi map. Section \ref{sec:pi and theta}
shows that it is unique with carrier $\Blin(\cH,\cK)$ and natural
in both $\cH$ and $\cK$. Section \ref{sec:C} shows how $\gC$ is
obtained by combining $\gJ$ and $\gT$, and proves its
just-mentioned properties.
\S\ref{sec:CV} brings us back to the point of departure,
namely the formulas (\ref{eq:usual Choi},\ref{eq:usual Choi inverse})
of the so-called Choi-Jamio{l}kowski isomorphism.
Normally, neither of these is \textit{derived} in any meaningful sense.
They arrive as inspired guesses and are then shown to satisfy what
is claimed of them.
Section \ref{sec:CV} shows that all Choi maps with carrier
$\cK\otimes\cH$ and natural in $\cK$ derive from $\gC$ via a unitary
$\cH\rightarrow \cH^*$. These are precisely the maps of the CJ Theorem,
and the formulas then follow as a consequence of naturality.
\S~\ref{sec:trace} discusses the trace properties and shows that
all $(\cH,\cK)$-Choi maps send trace-preserving maps to operators of
trace $\dim\cH$.
Proof of Prop.\ref{prop:Choi} is relegated to an appendix
because, while the proposition is important for context and motivation,
the proof has little direct r\^{o}le in anything else here.
One of the lemmas of the appendix is appealed to \S\ref{sec:CV}.

\subsubsection*{The perspective vis-\`{a}-vis the literature}

There are two distinct uses of bases in statements/version of the CJ Theorem. 
One is as a method of specifying a maximal entangled state vector in
$\cH\otimes\cH$.
This is a bad method, however, because the labelling of such vectors by ONBs of $\cH$ is both incomplete and redundant. In the framework we set up, there will be neither the slightest temptation toward, nor advantage in, such use of bases.
The other use is to facilitate computationally convenient matrix
representation. These become conflated in one popular way of writing
the Choi matrix of $\Lambda$, viz.
\begin{equation}\nonumber
  \sum_{i,j} \Dket{i}\Dbra{j} \otimes \Lambda(\Dket{i}\Dbra{j}),
\end{equation}
a form which goes back to Choi's original work\cite{Choi-75}.
Appealing though it may be on the surface, this eventuates in considerable
head-scratching about transformation properties, a problem which would not
arise if the uses of bases were kept separate.

Although the conventional CJ isomorphisms $\vJ_\Upsilon$
depends on a choice of $\Upsilon$, it seems that said choice is
very often a matter of total indifference, which fact encourages the
common habit of referring to 
``\textit{the} Choi-Jamio{l}kowski isomorphism'' (singular).
If so, one would be wise to ask whether it might be better to use
a version which required no such choice. As more complex 
applications emerge, as for instance to process tensors\cite{Pollock+18},
the greater the advantage in doing so.

As far as I can tell, the primary literature reflects little interest
in the kinds of questions offered above as motivation,
and the secondary literature even less.
Grabowski et al.\cite{Grabowski+Kus+Marmo-07} are interested in
doing away with the entangled state seed, under the slogan
``no use of coordinate (matrix) interpretation''.
Their work directly inspired the $\gJ$-transform used here.
Jiang et al.\cite{Jiang+Luo+Fu-13} discuss transformation properties
under change of $\Upsilon$. Broadly enough construed,
transformation properties are at the heart of this treatment.
This leads us to the use of a bit of category theory, but only as
an analysis tool. I see little connection with the
far more radical project of \textit{categorical quantum mechanics}
\cite{Abramsky+Coecke-09,Coecke+Kissinger-16a,Coecke+Kissinger-16b,Barnum+Duncan+Wilce-13}.

\section{Hilbert spaces}\label{sec:Hilbert}

The approach of this paper is based on the notion of
(finite-dimensional) Hilbert space, rather than on Dirac notation.
That may sound a strange contrast to make, so we will start with dual
spaces as Hilbert spaces, something Dirac notation makes extremely
awkward, at best. The general message is that the dual space,
tensor product, and linear-maps (again --- in finite dimensions) 
constructors yield Hilbert space outputs from Hilbert space
inputs, all in a coherent way.

Even the reader for whom the ideas here are very familiar should
probably skim this section for notation.
To help keep track of what space various vectors are in,
vectors in $\cH$ will be denoted by
$\hh$, $\hh'$, etc., and those in $\cK$ by $\kk$ etc.
Upper-case latin letters denote operators, upper-case greek
superoperators (when the distinction matters). 

\subsection{Dual spaces and transposition}

The dual space of $\cH$ is the vector space
\begin{equation}\nonumber
\cH^* = \Blin(\cH,\Cmplx)  
\end{equation}
of linear functionals on $\cH$.
We fix an inner product on $\cH^*$ by insisting that
the antilinear map $\hh \mapsto \pam{\hh}$,
where $\pam{\hh}(\hh') = \inpr{\hh}{\hh'}$, be antiunitary.
($\inpr{\hh}{\hh'}$ is an inner product evaluated at the pair
$({\hh},{\hh'})$, not a bra applied to a ket.)

This gives
\begin{equation}\label{eq:H* ip}
\inpr{\pam{\hh}}{\pam{\hh'}} \defeq \inpr{\hh'}{\hh}.   
\end{equation}


The \textbf{transpose} of $A\colon\Blin(\cH,\cK)$
is the operator $\tpose{A}$
in \mbox{$\Blin(\cK^*,\cH^*)$} defined by
\begin{equation}\label{eq:t def}
\inpr{\tpose{A}\pam{\kk}}{\pam{\hh}}
\defeq \inpr{A\hh}{\kk}.
\end{equation}
I.e., the matrix of $\tpose{A}$ with respect to the orthonormal bases
(ONBs) $\pam{\kk}_1,\ldots,\pam{\kk}_n$ and
$\pam{\hh}_1,\ldots,\pam{\hh}_n$ is the transpose of the matrix
of $A$ with respect to
${\hh}_1,\ldots,{\hh_n}$ and ${\kk}_1,\ldots,{\kk_n}$,
for \textit{any} choice of the latter bases.

The RHS of (\ref{eq:t def}) is equal to
$\inpr{\pam{\kk}}{\pam{A\hh}}$ by (\ref{eq:H* ip}), and to
$\inpr{\hh}{A^\dag\kk}$, which is
$\inpr{\pam\kk}{\tpose{(A^\dag)}\pam\hh}$.
Comparing these to our original expression
$\inpr{\tpose{A}\pam\kk}{\pam\hh}$ shows that
$\tpose{(A^\dag)}\pam\hh = (\tpose{A})^\dag\pam\hh = \pam{A\hh}$.
This motivates (already!) another definition
\begin{empheq}[box=\widefbox]{equation}
  \label{eq:star}
A^\ast \defeq \tpose{A}^\dag,  
\end{empheq}
such that $A^\ast \pam{\hh} = \pam{A\hh}$.
This operation seems not to have a standard name;
we call it \textbf{conjugation}.

A nice way to present the preceding equation is as the commutative diagram
\begin{equation}\label{eq:dual natural}
  \begin{tikzcd}
    \cH
    \arrow[d,rightarrow,"A"']
    \arrow[r,rightarrow,"\pam{(\;)}"]
    &
    \cH^*
    \arrow[d,rightarrow,"{A}^\ast"]
    \\
    \cK
    \arrow[r,rightarrow,"\pam{(\;)}"]
    & \cK^*
  \end{tikzcd}
\end{equation}
``Commutative'' simply means that, for a fixed start and finish node,
all compositions of arrows along a path from start to finish give the same
result. In this case, that strictly equational content is no more than what we
have already written down, immediately after (\ref{eq:star}).
Still, the left-vs-right and top-vs-bottom
arrangement of the diagram can give a clearer picture of what is going on.

\subsection{Rank-1 operators}

In $\Blin(\cH,\cK)$, the vector space of linear operators from
$\cH$ to $\cK$,
there is a special class of operators called \textit{dyadics} or
\textbf{rank-1 operators}. They are of the form $\kk\pam{\hh}$.  
We are going to manipulate these a little for practice with the notation.
The result of applying $\kk\pam\hh$ to $\hh'\in\cH$ is
$\kk \pam{\hh}(\hh') = \kk \inpr{\hh}{\hh'}$.
Next, $\inpr{\kk'}{\kk\pam{\hh}\cdot\hh'}
= \inpr{\kk'}{\kk} \inpr{\hh}{\hh'}
= \inpr{ \hh\pam{\kk}\cdot\kk' }{\hh'}$, 
so $(\kk\pam{\hh})^\dag = \hh\pam{\kk}$.
Finally,
$(\kk\pam\hh)^\ast \pam{\hh'}
= \pam{\kk\pam{\hh}\cdot \hh'}
=\pam{\kk}\pam{\inpr{\hh}{\hh'}}
=\pam{\kk}\inpr{\pam\hh}{\pam{\hh'}}$, so
$(\kk\pam{\hh})^\ast = \pam{\kk}\,\pam{\pam{\hh}}$.
Since $\ast = t \circ \dag$, we can immediately fill out 
the following diagram to conveniently keep track of these operations:
\begin{equation}\label{eq:bar-dag-star}
  \begin{tikzcd}
\kk\pam{\hh}\in\Blin(\cH,\cK)
    \arrow[r,rightarrow,"\dag"]
    \arrow[d,rightarrow,"\pam{{(\;)}}"']
    \arrow[dr,rightarrow,"\ast"]
& [-0.5em]
    \Blin(\cK,\cH) \ni \hh\pam{\kk}
    \arrow[d,rightarrow,"t"]
\\
\pam{\kk\pam{\hh}}\in \Blin(\cH,\cK)^*
    \arrow[r,dashrightarrow]
& 
    \Blin(\cH^*,\cK^*) \ni \pam{\kk}\, \pam{\pam\hh}
  \end{tikzcd}
\end{equation}
All the operations here are either unitary or antiunitary.
The $\pam{\pam\hh}$ at the lower-right node
may look a bit strange, but it is nothing new.
Our definition of the overbar says that
$\pam{\pam{\hh}}$ is the element of $\cH^{**}$ such that
$\pam{\pam{\hh}}(\pam{\hh'})
= \inpr{\pam{\hh}}{\pam{\hh'}}$.

%
\subsection{$\Blin(\cH,\cK)$ as a Hilbert space}\label{eq:L(H,K) as H-space}

As vector spaces, $\Blin(\cH,\cK)$ and $\cK\otimes\cH^*$
are \textit{naturally isomorphic} (``naturally'' will be clarified in
\S\ref{sec:T}), via
\begin{equation}\label{eq:T}
    \boxed{
  \begin{split}
  \gT \colon \Blin(\cH,\cK) & \xrightarrow{\cong} \cK\otimes\cH^*  \\
      \kk\pam{\hh} &\mapsto \kk\otimes\pam{\hh}.
  \end{split}}
 \end{equation}
It suffices to define $\gT$ only on rank-1 operators, 
because the thing we are defining is a linear operator,
and the rank-1 operators linearly span $\Blin(\cH,\cK)$.
In the sequel, we will use this technique of explicitly defining linear
operators only on a spanning set, and not even bother to mention the
linear extension.

Using the natural Hilbert space structure on $\cK\otimes\cH$,
namely
\begin{equation}\label{eq:tensor ip}
\inpr{\kk\otimes{\hh}}{\kk'\otimes{{\hh'}}} 
=
\inpr{\kk}{\kk'}
\inpr{{\hh}}{{\hh'}},
\end{equation}
we make $\Blin(\cH,\cK)$ a Hilbert space
by insisting $\gT$ be unitary, or a \textit{Hilbert isomorphism},
similarly to what was just done for $\cH^*$:
\begin{equation}\label{eq:HS derivation}
  \begin{split}
\inpr{\kk\pam{\hh}}{\kk'\pam{\hh'}}
    &\stackrel{(\ref{eq:T})}{=}
\inpr{\kk\otimes{\pam\hh}}{\kk'\otimes{\pam{\hh'}}} 
\\ &\stackrel{(\ref{eq:tensor ip})}{=}
      \inpr{\kk}{\kk'}
      \inpr{\pam{\hh}}{\pam{\hh'}}
\\  &\stackrel{(\ref{eq:H* ip})}{=}     
          \inpr{\kk}{\kk'}\inpr{\hh'}{\hh}
\\ &=     
          \Tr ({\hh\pam{\kk}} \circ {\kk'\pam{\hh'}})
\\  &\stackrel{(\ref{eq:bar-dag-star})}{=}     
      \Tr ({(\kk\pam{\hh})^\dag} \circ {\kk'\pam{\hh'}}). 
  \end{split}
\end{equation}
Therefore,
\begin{equation}\label{eq:HS}
  \inpr{A}{B} = \Tr A^\dag B.
\end{equation}
\subsection{An operator/vector identity}

The identity
\begin{equation}\label{eq:switch}
  \inpr{\kk\pam{\hh}}{A}  
  =
  \inpr{\kk}{A\hh} 
\end{equation}
will be repeatedly useful in the sequel. The inner products here
are, evidently, on $\Blin(\cH,\cK)$ (LHS) and $\cK$ (RHS).
This identity follows immediately from definition (\ref{eq:HS}),
but could also be taken as a definition of the Hilbert-Schmidt
inner product, with the LHS as \textit{definiendum}.

\subsection{Everything fits together}

\begin{itemize}
\item 
$\Cmplx$ is itself a Hilbert space over $\Cmplx$ with
inner product \mbox{$\inpr{a}{b} = \pam{a}b$}.
Here, the overbar is just complex conjugation, but
this is compatible with our other use of overbar.

\item
Considered under the guise $\Blin(\cH,\Cmplx)$, the dual space
$\cH^*$ has a Hilbert-Schmidt inner product.
We verify that it coincides with our earlier definition
(\ref{eq:H* ip}).
To do that, write $\pam{\hh}$ as $1\pam{\hh}$.
Now,
$\inpr{1\pam{\hh}}{1\pam{\hh'}}
  = \Tr (\hh\pam{1}\circ 1\pam{\hh'})
 = \Tr \hh\pam{\hh'}
 = \inpr{\hh'}{\hh}
 = \inpr{\pam{\hh}}{\pam{\hh'}}$,
as required.

\item
As linear spaces, $\cH$ and $\Blin(\Cmplx,\cH)$ are isomorphic
via $\hh \mapsto \hh\pam{1}$. 
Check that this is unitary:
$\inpr{\hh\pam{1}}{\hh'\pam{1}}
=\inpr{\hh}{\hh'}\inpr{\pam{1}}{\pam{1}}
=\inpr{\hh}{\hh'}$.

\item
Since $\hh\mapsto\pam{\hh}$ is antiunitary,
\mbox{$\hh \mapsto \pam{\pam{h}} \colon \cH \cong \cH^{**}$},
is a unitary map between $\cH$ and its bidual.
\end{itemize}

\section{The idea of naturality\cite{Eilenberg+MacLane-45,MacLane-71,%
  Pierce-91,Awodey-06,Riehl-16}}%
\label{sec:cats}

This paper contains many diagrams like the one in (\ref{eq:dual natural}),
as well as talk of such related categorial notions as
\textit{functors} and (especially) \textit{natural transformations},
so now may be a good time to start getting used to it.
We try to say here no more than is necessary for our purposes.

\subsection{Category}

The primary background \textbf{category} here is that of (finite-dimensional)
complex Hilbert spaces with all linear operators amonst them, which we call
$\mathsf{Hil}$. What does this mean?
A category is a structure $\mathsf{C}$ consisting of
\begin{itemize}
\item 
 a collection of \textbf{objects} $X,Y,\ldots$
\item for each ordered pair $(X,Y)$ of objects,
a set of \textbf{morphisms} $\mathsf{C}(X,Y)$
(said to have \textbf{domain} $X$ and \textbf{codomain} $Y$),
such that
\begin{enumerate}
\item 
$f\in\mathsf{C}(X,Y)$ and $g\in\mathsf{C}(Y,Z)$ are
composable to a morphism in $\mathsf{C}(X,Z)$ denoted $g\circ f$ or $gf$,
\item 
composition is associative,
\item 
for each object $X$, there is an identity
\mbox{$\mathrm{id}_X\in\mathsf{C}(X,X)$},
i.e., $\mathrm{id}_X\circ f = f$ or
$f \circ \mathrm{id}_X = f$ if $f$ has domain or codomain $X$, respectively.
\end{enumerate}
\end{itemize}

For typical categories, the objects are some kind of
structured set (e.g., Hilbert spaces) and morphisms are mappings which
respect that structure (e.g., linear mappings, or subclasses thereof).
In such a case, categorial composition is simply ordinary map composition.

A different example will help us get an intuitive feeling for
natural transformations.
A group can be construed as a category $\mathsf{G}$ having a single object,
say $*$; the group elements are the morphisms $\mathsf{G}(*,*)$.
Part of the group axioms are embodied by the
category axioms. We need add only the stipulation that every morphism is
an isomorphism, i.e., has an inverse.
From this perspective, we might say that categories are what
result when, starting from groups, you drop the requirement of
invertibility and allow the group multiplication to be partial.
(Almost --- this does not account for identity morphisms.)

\subsection{Functor}\label{sec:functor}

A \textbf{functor} $F$ from category $\mathsf{C}$ to $\mathsf{D}$ is a pair
of mappings (on objects and morphisms) which carries $\mathsf{C}(X,Y)$
into $\mathsf{D}(FX,FY)$ in such a way that the categorial structure
(composition and identity) is respected,
that is, $F(gf) = (Fg)(Ff)$ and $F\mathrm{id}_X = \mathrm{id}_{FX}$.

The map \mbox{$(A\colon\cH\rightarrow\cK)
  \,\mapsto\,
  (A^\ast\colon\cH^*\rightarrow\cK^*)$} of (\ref{eq:star})
is an example.
The demonstration is simple:
$(BA)^\ast\pam{\hh} = \pam{BA\hh}
= B^\ast(\pam{A\hh}) = B^\ast A^\ast \pam{\hh}$, and
$\Id_\cH^\ast = \Id_{\cH^*}$ is trivial.
The domain of $\ast$ is evidently
$\mathsf{Hil}$, but we might consider the codomain to be
$\mathsf{Hil}$, or the subcategory $\mathsf{Hil}^*$ of $\mathsf{Hil}$
containing only dual spaces,
or even the category of sets and set mappings, $\mathsf{Set}$.

Return to the example of a group regarded as a category $\mathsf{G}$.
A functor from $\mathsf{G}$ into the category $\mathsf{Vec}$ of
vector spaces and
linear maps consists of a vector space $\cV$ and a map of the group
elements to linear operators on $\cV$ which respects composition and
identity. In more conventional terms, such a functor is a
\textit{group representation}.
The idea of functor is thereby seen as a generalization of
group representation.

\subsection{Natural transformation}

Finally, we come to the notion of \textbf{natural transformation},
which is what we really need out of this digression.
Sticking with our group example, suppose for concreteness that
$\mathsf{G}$ is the Euclidean group.
Take the functor $F$ to represent the $\mathsf{G}$-action on
a space $\cQ$ of electric charge distributions,
and a second functor $F'$ to be the $\mathsf{G}$-action on
a space $\cE$ of electric field configurations.
Then, the diagram
\begin{equation}\nonumber
  \begin{tikzcd}
    \cQ
    \arrow[r,rightarrow,"\alpha"]
    \arrow[d,rightarrow,"{F}g"]
    &
    \cE
    \arrow[d,rightarrow,"F'g"]
    \\
    \cQ
    \arrow[r,rightarrow,"\alpha"]
    &
    \cE
  \end{tikzcd}
\end{equation}
commutes if $\alpha$ correctly implements Coulomb's law:
the electric field of a rotated/translated charge distribution
is the same as the rotated/translated electric field of the original
charge distribution.
Any $\alpha$ which is a morphism in $\mathsf{Vec}$ and makes the
diagram commute is a \textbf{natural transformation} in this situation.

To generalize, we consider two functors
$F,G\colon\mathsf{C}\rightarrow\mathsf{D}$ with the same domain and
codomain categories. A natural transformation consists of a
$\mathsf{D}$-morphisms $\alpha_{_X}$ for each $\mathsf{C}$-object $X$,
such that the diagram
\begin{equation}\nonumber
  \begin{tikzcd}
    FX
    \arrow[r,rightarrow,"\alpha_{_X}"]
    \arrow[d,rightarrow,"{F}f"]
    &
    GX
    \arrow[d,rightarrow,"Gf"]
    \\
    FY
    \arrow[r,rightarrow,"\alpha_{_Y}"]
    &
    GY
  \end{tikzcd}
\end{equation}
commutes for any $\mathsf{C}$-morphism $f\colon X\rightarrow Y$.
One says that $\alpha_{_X}$ is the $X$-component of $\alpha$; the
diagram above is a \textbf{naturality square}; and the notation
$\alpha\colon F \Rightarrow G$ is common to indicate a natural
transformation.
In the main body of the paper, we will refrain from writing indices
on our natural transformations, to avoid excessive clutter.

Look back at diagram (\ref{eq:dual natural}).
Is $\pam{(\;)}$ a natural transformation between the identity functor
and $\ast$? Well, it (really, its components) needs to be in the target
category.
One way to make the claim proper is to declare the codomain
category to be $\mathsf{Hil}$ with antilinear maps added to the morphisms,
say $\mathsf{Hil}^+$. 
Then $\pam{(\;)}$ is a natural transformation, indeed a natural isomorphism,
between the embedding of $\mathsf{Hil}$ into $\mathsf{Hil}^+$
and $\ast$, also viewed as a functor into the bigger category.
Alternatively, we could consider the codomain to be $\mathsf{Set}$,
the category of sets and all set functions. 

\section{$\gT$ is a natural unitary}\label{sec:T}

This section shows that $\gT$ (\ref{eq:T}) is,
up to a phase, the unique unitary from
$\Blin(\cH,\cK)$ to $\cK\otimes\cH^*$, natural in $\cH$ and $\cK$.
Both the results and the techniques will be of use later.

\subsection{Bifunctors}

A generic pair of arrows
$A\colon \cH\rightarrow\cM$ and
$B\colon \cK\rightarrow\cN$ in $\mathsf{Hil}$ can be considered
as an arrow
\begin{equation}
\begin{tikzcd}
    (\cH,\cK)
    \arrow[d,rightarrow,"{(A,B)}"]
\\
    (\cM,\cN)
  \end{tikzcd}
\end{equation}
of a new category, $\mathsf{Hil}\times\mathsf{Hil}$. 
This is a purely formal construction, and has no meaning beyond
the syntax. We can try to interpret it by Hilbert space operators
via a functor into $\mathsf{Hil}$.
Here are the images of the above arrow under three potentially
useful functors:
    \begin{equation}\label{eq:3 functors}
  \begin{tikzcd}
    \cK\otimes\cH
  \arrow[d,rightarrow,"B\otimes {A}"']
    &[-1.5em]
  \Blin(\cH,\cK)
   \arrow[d,rightarrow,"B\blank A^\dag"']
    &[-1.5em]
 \cK\otimes\cH^*
  \arrow[d,rightarrow,"B\otimes {A^\ast}"']
    \\
  \cN\otimes\cM
&  \Blin(\cM,\cN)
&  \cN\otimes\cM^*
  \end{tikzcd}
\end{equation}
The first one, $B\otimes A$, probably suggests itself most readily, but
it is the other two which are of most interest to us.
One also calls these \textbf{bifunctors}, as they can be considered
functors of two variables, and one can consider transformations which
are natural in one or both variables.
We use the term \textbf{binatural} to emphasize naturality in both, in
contexts where that might not be clear.

\subsection{naturality square}

The diagram at right here,
\begin{equation}\label{eq:T natural}
  \begin{tikzcd}
    (\cH,\cK)\arrow[d,rightarrow,"{(A,B)}"']
    \\
    (\cM,\cN)
  \end{tikzcd}
  \quad
\begin{tikzcd}
  \Blin(\cH,\cK)
  \arrow[d,rightarrow,"B\blank A^\dag"']
  \arrow[r,rightarrow,"\gT","\cong"']
  & [-0.5em]
  \cK\otimes\cH^*
  \arrow[d,rightarrow,"B\otimes {A^\ast}"]
  \\
  \Blin(\cM,\cN)
  \arrow[r,rightarrow,"\gT"',"\cong"]
  &
  \cN\otimes\cM^*
\end{tikzcd}
\end{equation}
\noindent
is the naturality square for $\gT$.
Here is the commutativity check:
\begin{equation}\nonumber
\begin{tikzcd}
  \kk\pam\hh 
  \arrow[d,mapsto,"B\blank A^\dag"']
  \arrow[r,mapsto,"\gT"]
  & [-0.5em]
  \kk\otimes\pam\hh
  \arrow[d,mapsto,"B\otimes {A^\ast}"]
  \\
  (B\kk)(\pam\hh A^\dag)
  \arrow[r,mapsto,"\gT"']
  &
  B\kk \otimes A^\ast\pam\hh
\end{tikzcd}
\end{equation}
Going across the bottom uses $A^\ast\pam{\hh}=\pam{A\hh} = \pam{\hh}A^\dag$.
This shows that $\gT$ is a natural transformation, the components of which
are unitary, thus we call it a \textbf{natural unitary}.
The notational convention of \S\ref{sec:cats} would require an
index on $\gT$ to indicate which component is involved; we will dispense
with that. Instead, we think of it as a \textbf{polymorphic} object
(to borrow a computer programming term). Which particular component is
required is determined by the argument. Notice that, although each compoent
is unitary, $\gT$ as a whole does not have a well-defined inverse.
$\gT$ maps both $\Blin(\cH\otimes\cM,\cK)$ and $\Blin(\cM,\cK\otimes\cH^*)$ 
onto to $\cK\otimes\cH^*\otimes\cM^*$.

To avoid clutter, we have written ``$\cong$'' in the diagram
(\ref{eq:T natural}) instead of more accurately labelling the
arrows as unitary. This will be a habit.

\subsubsection*{No natural transformation to the other functor}

We said that the first functor in (\ref{eq:3 functors}) was not
of interest. That is because there is no natural transformation to
it from either of the others.
To see why, consider the following maybe-diagram.
\begin{equation}\label{eq:unnatural}
\begin{tikzcd}
  \Blin(\cH,\cK)
  \arrow[d,rightarrow,"B\blank A^\dag"']
  \arrow[r,dashrightarrow,"?"]
  & [-0.5em]
  \cK\otimes\cH
  \arrow[d,rightarrow,"B\otimes {A}"]
  \\
  \Blin(\cM,\cN)
  \arrow[r,dashrightarrow,"??"]
  &
  \cN\otimes\cM
\end{tikzcd}
\end{equation}
No linear or antilinear arrows ``?'' and ``??'' can make this commute
for all $B\colon\Blin(\cK,\cN)$ and $A\colon\Blin(\cH,\cM)$ because
multiplying $A$ and $B$ by $i$ leaves the left arrow unchanged, but 
reverses the sign of the right arrow.

\subsection{Yoneda formula}

We can compute $\gT(BA^\dag)$ for any
\begin{equation}\nonumber
  \begin{tikzcd}
    \cM\arrow[r,->,"A^\dag"]
    & \cH\arrow[r,rightarrow,"B"]
    & \cN
  \end{tikzcd}
\end{equation}
using the naturality square with $\cK=\cH$ like so,
\begin{align}
  \gT(BA^\dag)
   & = \gT\circ(B\blank A^\dag)(\Id_\cH) \nonumber
    \\ & \stackrel{(\ref{eq:T natural})}{=}
         (B\otimes A^\ast) \gT(\Id_\cH) \nonumber
\\ & =
         (B\otimes A^\ast) I_\cH.
  \label{eq:T Yoneda}
\end{align}
The last line incorporates the definition
\begin{equation}\label{eq:T(Id)}
  \begin{split}
  I_\cH & \defeq \gT \Id_\cH\colon\cH\otimes\cH^*
\\ & = \onbsum{i} \gT(\hh_i\pam{\hh_i})
\\ & = \onbsum{i} \hh_i\otimes\pam{\hh_i}.
  \end{split}
\end{equation}

Eq. (\ref{eq:T Yoneda}) is essentially an instance of the Yoneda Lemma;
after all, $(\cH,\cK) \mapsto \Blin(\cH,\cK)$ is a Hom bifunctor.
We will not actually use that result, but, for convenience and as a
nod to the connection, 
formulas such as (\ref{eq:T Yoneda})
which express the action of a natural transformation on arbitrary
operators in terms of its action on an identity will be called
\textbf{Yoneda formulas}. The main contrast here is to formulas,
such as that in the definition (\ref{eq:T}),
which explains the action of $\gT$ in terms of some kind of decomposition
of the argument (there as rank-1 operators). These will be called
\textbf{structural formulas}.
Returning to (\ref{eq:T Yoneda}), since
$B\otimes A^\ast$ equals $(B\otimes\Id )(\Id \otimes A^\ast)$, those
factors commute, and either one can be ``pushed back inside'' by reading
(\ref{eq:T Yoneda}) backward, we have the variants (also to be called
\textit{Yondeda formulas})
\begin{equation}\nonumber
  \begin{split}
  \gT(BA^\dag)
& =  (B\otimes \Id_{\cM^*}) \gT(A^\dag)
\\ & =  (\Id_{\cN} \otimes A^\ast) \gT(B).
  \end{split}
\end{equation}

\subsection{Uniqueness}\label{sec:T unique}

The reader may easily check that the first two lines of (\ref{eq:T Yoneda})
use only that $\gT$ is a natural transformation, i.e.,
the naturality square commutes.
We can go further with this line of argument to show that
any natural transformation
\mbox{$\alpha\colon B \blank A^\dag \Rightarrow B \otimes A^\ast$}
is determined by $\alpha(\Id_\Cmplx)$, and that if $\alpha$ is
natural unitary, then the freedom here amounts to only a phase,
so that $\alpha = c\gT$ for some modulus-1 complex number $c$.
Here is the relevant instance of the naturality square:
\begin{equation}\label{eq:T unique}
\begin{tikzcd}
  \Blin(\Cmplx,\Cmplx)
  \arrow[d,rightarrow,"(\hh\pam{1})\blank 1\pam{\hh}"']
  \arrow[r,rightarrow,"\alpha"]
  & [-0.5em]
  \Cmplx\otimes\Cmplx^*
  \arrow[d,rightarrow,"\hh\pam{1}\otimes \pam{\hh}\,\pam{\pam{1}}"]
  \\
  \Blin(\cH,\cH)
  \arrow[r,rightarrow,"\alpha"]
  &
  \cH\otimes\cH^*
\end{tikzcd}
\end{equation}
The point is that $\Cmplx\otimes\Cmplx^*$ is
one-dimensional, so $\alpha(\Id_\Cmplx) = c(1\otimes\pam{1})$,
and therefore $\alpha(\hh\pam{\hh}) = c\hh\otimes\pam{\hh}$.
Since $\alpha$ must be linear, summing over a basis gives
$\alpha(\Id_\cH) = c I_\cH$, with $I_\cH$ from (\ref{eq:T(Id)}).
The Yoneda formula determines the action of $\alpha$ on everything else.

\subsection{Operator/state-vector duality}
\label{sec:operator/state-vector duality}

Given a (not necessarily orthogonal) basis $(\hh_n)$ of $\cH$,
state vectors in $\cK\otimes\cH$ are in one-to-one correspondence with
sums
\begin{equation}
\sum_n \kk_n\otimes\hh_n,
\end{equation}
i.e., choices of $(\kk_n)$. The map $\hh_n \mapsto \kk_n$
extends by linearity to an operator in $\Blin(\cH,\cK)$.
Abstracting over state vectors, this is a linear operator
$\cK\otimes\cH \rightarrow \Blin(\cH,\cK)$.
However, \textit{this correspondence depends on the basis $(\hh_n)$}.
This is essentially the point of the impossibility of
commutative diagram (\ref{eq:unnatural}), in other terms.

Our natural unitary $\gT$ shows how to do better.
Indeed, instantiated at $\Blin(\cH^*,\cK)$, $\gT$
is a unitary mapping onto $\cK\otimes\cH^{**}$,
which is tantamout to a unitary onto $\cK\otimes\cH$,
since $\cH^{**}$ is naturally identified with $\cH$.

More concretely, given
$A$ in $\Blin(\cH^*,\cK)$,
then
\begin{equation}
  \begin{split}
\gT(A)
&=  (A \otimes \Id)I_{\cH^*}
    \\ & =
  \onbsum{n} A \pam{\hh}_n \otimes \pam{\pam{\hh}}_n\colon \cH\otimes\cH^{**}.
  \end{split}
\end{equation}
Although this formula involves a particular ONB $(\hh_n)$, the state
vector $\gT(A)$ is independent of that choice.
Moreover, the norm of $\gT(A)$ is the same as that of $A$.

Thus, we have an operator/state-vector duality.
Lifting $\gT$ via functor $\btheta$ 
(introduced in \S\ref{sec:pi and theta})
will be a key part of the channel/state duality derived later.

\section{$\gJ$-transform}\label{sec:J}

$\Blin(\cH,\cK)$ is a Hilbert space whenever $\cH$ and $\cK$ are so,
hence if $\cM$ and $\cN$ are Hilbert spaces, so is
$\Blin(\Blin(\cH,\cK),\Blin(\cM,\cN))$.
This section constructs a unitary mapping $\gJ$ of 
$\Blin(\Blin(\cH,\cK),\Blin(\cM,\cN))$
onto
$\Blin(\Blin(\cH,\cM),\Blin(\cK,\cN))$, in a way 
uniform for all quadruples of Hilbert spaces.
Although the basic development is just as easy with four potentially
distinct Hilbert spaces,
it is the specialization
$\Blin(\Blin(\cH),\Blin(\cK)) \bijam \Blin(\Blin(\cH,\cK))$
which is of ultimate interest, because it is a Choi map,
as demonstrated in \S\ref{sec:J is a Choi}.
Not surprisingly, it is also natural, as the next section shows.

\subsection{Definition}

$\gT$ induces a unitary
\begin{equation}\nonumber
  \Blin(\Blin(\cH,\cK),\Blin(\cM,\cN))
  \cong
  \cN\otimes\cM^*\otimes\cK^*\otimes\cH^{**}
\end{equation}
via application once at the outer level and twice at
the inner level.
Since permuting the order of factors
in a tensor product (in particular, swapping $\cM^*$ and $\cK^*$)
is a unitary, we immediately deduce that
\begin{align}
\gJ \colon  
  \Blin(\Blin(\cH,\cK),\Blin(\cM,\cN))
& \xrightarrow{\cong}  \Blin(\Blin(\cH,\cM),\Blin(\cK,\cN))
\nonumber    \\
\nn\pam{\mm}\circ\pam{\kk\pam{\hh}}
& \mapsto \nn\pam{\kk}\circ\pam{\mm\pam{\hh}}
  \label{eq:jam}
\end{align}
is also unitary. That is, for each choice of $\cH,\cK,\cM,\cN$ we
get a specific unitary map.

\subsection{Key properties of $\gJ$}\label{sec:key props}

At a slightly higher level of description, 
rank-1 elements of $\Blin(\Blin(\cH,\cK))$ have the form
$A\pam{B}$, with \mbox{$A,B\in\Blin(\cH,\cK)$}.
The linear span of these is all of $\Blin(\Blin(\cH,\cK))$,
and we can use them to give a neat structural formula for
$\gJ$ using the notation
\begin{empheq}[box=\widefbox]{equation}\label{eq:blank}
A\blank B \defeq \; X \mapsto AXB
\end{empheq}
for $A\in\Blin(\cK)$, $B\in\Blin(\cH)$.
($\blank$ here is merely a syntactic device indicating where the
argument goes.)
With this, we have the following key properties:
\begin{subequations}\label{eq:key props}
  \begin{empheq}[box=\widefbox]{align}
  \text{polymorphic } \gJ \text{ is its own inverse}
  \label{eq:own inverse} \\
  \text{instantiations of } \gJ \text{ are unitary}
  \label{eq:J unitary} \\
  A\pam{B}  \bijam  A \blank B^\dag
  \label{eq:J rank-1} \\
 \gJ(\Lambda\otimes\Gamma) = \gJ \Lambda \otimes \gJ \Gamma
    \label{eq:jam distributes over otimes}
  \end{empheq}
  \end{subequations}
In fact, either direction of (\ref{eq:J rank-1}) is adequate to define $\gJ$.
However, it is not obvious from that form,
as it is from our official low-level definition,
either that $\gJ$ is self-inverse, or that it is a Hilbert isomorphism.

\begin{proof}
[Proof of \paren{\ref{eq:own inverse},\ref{eq:J unitary}}] 
  Immediate from the definition.
\end{proof}

\begin{proof}[Proof of \paren{\ref{eq:J rank-1}}]
  By linearity, it suffices to consider the case of
  rank-1 $A$ and $B$. To wit,
\begin{equation}\nonumber 
\begin{split}
\Big(\gJ ( \nn\pam{\kk}\circ\pam{\mm\pam{\hh}} )\Big) X
\stackrel{(\ref{eq:jam})}{=} \, &
            \Big( \nn\pam{\mm}\, \circ \pam{\kk\pam{\hh}} \Big) X
  \\ = \, &
            \nn\pam{\mm} \inpr{\kk\pam{\hh}}{X}
  \\ \stackrel{(\ref{eq:switch})}{=} \, &
            \nn \inpr{\kk}{X\hh} \pam{\mm}
  \\ \stackrel{(\ref{eq:bar-dag-star})}{=} \, &
            (\nn \pam{\kk}) X (\mm\pam{\hh})^\dag
\end{split}
\end{equation}
\end{proof}
\begin{proof}[Proof of \paren{\ref{eq:jam distributes over otimes}}]
\begin{equation}\nonumber
\begin{split}
  A\pam{B} \otimes S\pam{T}
 = & (A\otimes S) (\pam{B\otimes T}) 
  \\ \overset{(\ref{eq:J rank-1})}{\underset{\gJ}{\leftrightarrow}}
  & 
    (A\otimes S) \blank (B^\dag \otimes T^\dag)
  \\ = &
       A\blank B^\dag \otimes S \blank T^\dag
  \\  \stackrel{(\ref{eq:J rank-1})}{=} &
    \gJ(A\pam{B}) \otimes \gJ(S \pam{T}).
\end{split}
\end{equation}
\end{proof}

\subsection{$\gJ$ is a Choi map}\label{sec:J is a Choi}

Now we specialize to the case $\cM=\cH$ and $\cK=\cN$, i.e.,
\begin{equation}\label{eq:J B(B(H),B(K))}
\Blin(\Blin(\cH),\Blin(\cK))
\bijam  
\Blin(\Blin(\cH,\cK)),
\end{equation}
which is the one of primary interest for quantum physics.
This section proves that this is a Choi map in four steps.
Henceforth, ``$\gJ$-transform'' refers to this two-variable
specialization by default.
\subsubsection*{Monotonicity, complete positivity, and positivity}
First, we need to establish some terminology.
\mbox{$\Lambda\in\Blin(\Blin(\cH,\cK))$} is \textbf{rank-$N$-positive}
if \mbox{$\inpr{A}{\Lambda A}\ge 0$} for all $A$ of rank (denoted $\rk$)
not exceeding $N$.
\mbox{$\Lambda\in\Blin(\Blin(\cH),\Blin(\cK))$}
is \textbf{monotone}
if $\Lambda A \ge 0$ whenever $A \ge 0$; it is
\textbf{$N$-monotone}
if \mbox{$\Lambda\otimes\Id_{\Blin(\Cmplx^N)}$} is monotone.
So, $1$-monotone is synonymous with monotone;
$\CP$ with ``$\forall N$, $N$-monotone'';
and ordinary operator positivity with ``$\forall N$, rank-$N$-positive''.

The set of $\CP$ maps in $\Blin(\Blin(\cH),\Blin(\cK))$ will
be denoted $\CP(\cH,\cK)$. So, we aim to show that $\gJ$ maps
$\CP(\cH,\cK)$ bijectively onto the cone $\Blin(\Blin(\cH,\cK))^+$ of positive
operators on $\Blin(\cH,\cK)$.
In fact, we will prove the stronger assertion that $\Lambda$ is
$N$-monotone iff $\gJ \Lambda$ is rank-$N$ positive.

\subsubsection{A formula for $\gJ$}

The new notation
\begin{empheq}[box=\widefbox]{equation}\label{eq:pi}
    \bpi(\hh) \defeq \hh\pam{\hh},
\end{empheq}
will be very useful.
In Dirac notation, this would be written $\Dket{\hh}\Dbra{\hh}$.
Note that
\begin{equation}
\bpi(\hh)^* = \bpi(\pam{\hh}) = \tpose{(\bpi(\hh))},
\end{equation}
as follows from (\ref{eq:bar-dag-star}).

By use of $\bpi$, the instance (\ref{eq:J B(B(H),B(K))}) of $\gJ$
is fully expressed by
\begin{equation}\label{eq:pi pi pi}
\bpi(\kk) \pam{\bpi(\hh)} \bijam \bpi(\kk\pam{\hh}).
\end{equation}
This is immediately verified by expanding the LHS as
$\kk\pam{\kk}\circ \pam{\hh\pam{\hh}}$
and the RHS as
$(\kk\pam{\hh})\circ(\pam{\kk\pam{\hh}})$.

\subsubsection{Monotone  $\bijam$  rank-1 positive}

\begin{equation}\nonumber
  \begin{split}
  \inpr{\bpi(\kk)}{\Lambda\cdot\bpi(\hh)}
& = 
  \inpr{\bpi(\kk)\pam{\bpi(\hh)}}{\Lambda}
\\ & \stackrel{(\ref{eq:pi pi pi})}{=}
  \inpr{\bpi(\kk\pam{\hh})}{\gJ\Lambda}
\\ & {=}
  \inpr{\kk\pam{\hh}}{\gJ\Lambda\cdot\kk\pam\hh}
  \end{split}
\end{equation}
Because
$\inpr{\kk}{(\Lambda\cdot\bpi(\hh))\kk}
=\inpr{\bpi(\kk)}{\Lambda\cdot\bpi(\hh)}$,
the LHS is positive for all $\kk$ and $\hh$ iff
$\Lambda$ is monotone, similarly the RHS is so iff 
$\gJ\Lambda$ is rank-1 positive.

\subsubsection{ $\Lambda\otimes{\Id}_{\Blin(\cN)}  %
\bijam \gJ\Lambda\otimes{\pibox}(\Id_{\cN})$ }

This follows from distributivity of $\gJ$ over $\otimes$,
and the case
$\Id_{\Blin(\cN)} = \Id_\cN\blank\Id_\cN \bijam \Id_\cN\pam{\Id_\cN}
= \bpi(\Id_\cN)$ of (\ref{eq:J rank-1}).

\subsubsection%
[$\Gamma\otimes\pibox({\Id}_{\Cmplx^N})$  is rank-1 positive 
\\ $\Leftrightarrow \Gamma$ is rank-$N$ positive ]
{$\Gamma\otimes\pibox({\Id}_{\Cmplx^N})$  is rank-1 positive 
  $\;\Leftrightarrow\;$  $\Gamma$ is rank-$N$ positive }

First, note
\begin{equation}\nonumber
  \pam{\Id}_{\cH}\cdot\hh'\pam{\hh}
  =  \inpr{\Id_{\cH}\hh}{\hh'}
  =  \inpr{\hh}{\hh'}.
\end{equation}

Now, we use a purification technique.
If
\begin{equation}
B = \sum_{i=1}^n{\kk_i}\pam{\hh}_i,
\end{equation}
let $\ee_1,\ldots,\ee_n$ an ONB of $\Cmplx^n$,
\begin{equation}\nonumber
  \begin{split}
  {\kk}_\circ & \defeq \sum_{i=1}^n \kk_i\otimes\ee_i,
    \\
  {\hh}_\circ & \defeq \sum_{i=1}^n \hh_i\otimes\ee_i,
\\ \widetilde{B} & \defeq \kk_\circ\pam{\hh}_\circ.
  \end{split}
\end{equation}
With these definitions,
\begin{equation}\label{eq:purify 1}
\begin{split}
  (\Gamma \otimes 1\pam{\Id}_{\Cmplx^n})\widetilde{B}
  &=
 \sum_{ij} (\Gamma \otimes 1\pam{\Id})
    \kk_i\pam{\hh}_j\otimes \ee_i\pam{\ee}_j
\\ & = \Gamma B \otimes 1 \colon\Blin(\cK)\otimes\Cmplx.
\end{split}
\end{equation}
We put in an explicit tensor factor $\Cmplx$
because the next step is clearer with it.
Next, using (\ref{eq:purify 1}) twice,
\begin{equation}\label{eq:jam choi step 3}
  \begin{split}
  \inpr{B}{\Gamma B}
  & =
  \inpr{(\Id\otimes 1\pam{\Id})\cdot \widetilde{B}}%
  {(\Gamma \otimes 1\pam{\Id}_{\Cmplx^n})\cdot \widetilde{B}}
\\  & =
  \inpr{\widetilde{B}}%
  {(\Gamma \otimes \bpi(\Id_{\Cmplx^n})) \widetilde{B}}.
  \end{split}
\end{equation}
The second equality results from flipping $\Id\otimes 1\pam{\Id}$
to the other side, using $(1\pam{\Id})^\dag(1\pam{\Id}) = \bpi(\Id)$.
Since $B$ here is free, apart from having rank not exceeding $N$,
(\ref{eq:jam choi step 3}) is equivalent to:
$\Gamma\otimes\bpi({\Id}_{\Cmplx^N})$  is rank-1 positive 
  $\;\Leftrightarrow\;$  $\Gamma$ is rank-$N$ positive.

  \subsubsection%
  [ $\Lambda$ is $N$-monotone  \\
  $\Leftrightarrow \gJ\Lambda$ is rank-$N$ positive]
  { $\Lambda$ is $N$-monotone  $\Leftrightarrow$
  $\gJ\Lambda$ is rank-$N$ positive}

Finally, we obtain the promised conclusion
by chaining together the preceding three steps:
$\Lambda$ is $N$-monotone  $\equiv$ 
$\Lambda\otimes\Id_{\Blin(\Cmplx^N)}$ is monotone
$\Leftrightarrow$
\mbox{$\gJ\Lambda\otimes\pam{\Id}_{(\Cmplx^N)}$} is rank-$1$ positive
$\Leftrightarrow$
$\gJ\Lambda$ is rank-$N$ positive.

\section{Life of \pibox\ and \thetabox}\label{sec:pi and theta}

We can achieve a very tidy formulation of the $\gJ$ transform
using the definition
  \begin{equation}\label{eq:theta}
    \boxed{
      \begin{split}
\Blin(\cH,\cK) & \xrightarrow{\btheta} \Blin(\Blin(\cH),\Blin(\cK))
  \\
A & \mapsto A\blank A^\dag 
  \end{split}}
\end{equation}
Indeed, with the new notation, (\ref{eq:J rank-1}) immediately
specializes to
\begin{equation}\label{eq:J theta pi}
\btheta(A) \bijam \bpi(A).  
\end{equation}
This is just the tip of the iceberg, however. This section explores the
way $\bpi$ and $\btheta$ work together, and gives the promised demonstration
that $\gJ$ is a natural unitary, in fact the \textit{unique} natural Choi map
with carrier $\cK\otimes\cH^*$.

\subsection{$\gJ$ in terms of \pibox\ and \thetabox}\label{sec:J via theta-pi}

The import of the preceding section is effectively expressed by
the commutative diagram
\begin{equation}\label{eq:central diagram}
  \begin{tikzcd}
    & [-3em]
    \Blin(\cH,\cK)
      \arrow[dr,rightarrow,"\bpi"]
      \arrow[dl,rightarrow,"\btheta"']
    & [-3em]
    \\
    \CP(\cH,\cK)
    \arrow[rr,leftrightarrow,"\gJ"]
    \arrow[d,hookrightarrow,"\subset"]
    && 
    \Blin(\Blin(\cH,\cK))^+
    \arrow[d,hookrightarrow,"\subset"]
    \\ 
    \hspace{-2em}
    \Blin(\Blin(\cH),\Blin(\cK))
 \arrow[rr,leftrightarrow,"\gJ"]
 &&
 \Blin(\Blin(\cH,\cK))
\end{tikzcd}
\end{equation}

$\bpi[\Blin(\cH,\cK)]$ is not
all of $\Blin(\Blin(\cH,\cK))^+$, but it does generate the latter as
a convex cone. That is, any positive operator $\Gamma$ on
$\Blin(\cH,\cK)$ can be written as $\sum\bpi(A_n)$ for suitable $A_n$.
For instance, an eigenvector decomposition will do (after all, 
$\Blin(\cH,\cK)$ is just another Hilbert space).
Consequently, $\gJ\Gamma = \sum\btheta(A_n)$.
Since $\gJ$ is bijective, it follows that every $\CP$ map 
has such \textbf{Kraus decompositions}.

The diagram (\ref{eq:central diagram}) should be regarded as being in
$\mathsf{Set}$, for a couple of reasons.
Most importantly, $\bpi$ and $\btheta$ are neither linear nor
antilinear. These are the only maps of that sort that we shall deal with.
Also, $\CP(\cH,\cK)$ and $\Blin(\Blin(\cH,\cK))^+$ are not vector spaces.
But, $\CP(\cH,\cK) \bijam \Blin(\Blin(\cH,\cK))^+$ has a unique linear
extension shown as the bottom row. This is worth pointing out here because
many of the maps in later diagrams have that property; they usually
bear a reading as representions of their linear extensions.
E.g., we can regard 
\mbox{$\gJ\colon \CP(\cH,\cK) \rightarrow \Blin(\Blin(\cH,\cK))^+$}
as an arrow in $\mathsf{Hil}$ even though, on the face of it, it 
is not so.

\subsection{\thetabox\ is a monoidal functor from $\mathsf{Hil}$ to $\CP$}
\label{sec:monoidal functor}

$\btheta$ is a functor:
Clearly $\btheta(\Id_{\cH}) = \Id_{\Blin(\cH)}$, and
$\btheta(AB) = AB\blank B^\dag A^\dag$ carries $X$ into
$ABXB^\dag A^\dag = \btheta(A)\btheta(B)X$.
Furthermore, \mbox{$\btheta(A)\cdot B \ge 0$} whenever
\mbox{$B\ge 0$},
i.e., $\btheta(A)$ is monotone.
So, we can write it as
\begin{equation}\label{eq:theta functor}
  \begin{tikzcd}
    \cH\arrow[d,rightarrow,"A"] \\ \cK
  \end{tikzcd}
\quad  \mapsto
  \begin{tikzcd}
\Blin(\cH)^+ \arrow[d,rightarrow,"\btheta(A)"] \\ \Blin(\cK)^+
  \end{tikzcd}
\end{equation}
Moreover, $\btheta$ is
\textbf{monoidal} in the sense that
\begin{equation}
\btheta(A\otimes B) = \btheta(A)\otimes\btheta(B).
\end{equation}
Therefore, $\btheta(A)$ is not just monotone, but $\CP$
[Proof:
$\btheta(A\otimes\Id_\cN)$ is monotone, but
$\btheta(\Id_\cN) = \Id_{\Blin(\cN)}$, so
$\btheta(A\otimes\Id_\cN) \\ = \btheta(A)\otimes\Id_{\Blin(\cN)}$.]

The reader may be feeling some unease because the codomain category
of the functor $\btheta$ was not specified. Taking (\ref{eq:theta functor})
at face value, it would seem to be the category $\CP$, the morphisms of which
are $\CP$ maps between cones of positive operators. Are the objects of
$\CP$ cones of positive operators such as $\Blin(\cH)$?
Well, they could be. But we could also take ``that object'' to be the underlying
Hilbert spaces $\cH$. This choice would be more consistent
with the notation $\CP(\cH,\cK)$. For that matter, we could even suppose
it to be $\Blin(\cH)$. It does not matter than much. What is required of
objects is that they unambiguously label domains and codomains of morphisms.
With the last option mentioned ($\Blin(\cH)$), we can suppose that the
codomain category is $\mathsf{Hil}$ itself; $\btheta(A)$ happens to be
$\CP$, so we write it as in (\ref{eq:theta functor}) to emphasize this.
A more sophisticated way to cast that is to recognize that we can
\textbf{embed} $\CP$ as a \textbf{subcategory} of $\mathsf{Hil}$
whenever convenient. This last point of view is a good one for the
purposes of this paper.

\subsection{And \pibox\ is a natural transformation}\label{sec:pi natural}

This is expressed by the diagram (in $\mathsf{Set}$)
\begin{equation}\label{eq:pi natural}
  \begin{tikzcd}
    \cH
    \arrow[d,rightarrow,"A"']
    \arrow[r,rightarrow,"\bpi"]
    &
    \Blin(\cH)^+
    \arrow[d,rightarrow,"\btheta(A)"]
    \\
    \cK
    \arrow[r,rightarrow,"\bpi"']
    & \Blin(\cK)^+
  \end{tikzcd}
\end{equation}
with the commutativity check
\begin{equation}\nonumber
  \begin{split}
  \btheta(A)(\bpi(\hh))
  & = A(\hh\pam{\hh})A^\dag  
\\ & = (A\hh)(\pam{\hh}A^\dag)  
\\ & = (A\hh)\pam{A\hh}  
\\ & = \bpi(A\hh).\quad \checkmark
  \end{split}
\end{equation}

\subsection{\pibox\ is a special case of \thetabox}\label{sec:preparation channel}

The (unique) instantiation of (\ref{eq:pi natural}) at
$\Cmplx$ at upper-left,
\begin{equation}\nonumber
  \begin{tikzcd}
    \Cmplx
    \arrow[d,rightarrow,"\hh\pam{1}"']
    \arrow[r,rightarrow,"\bpi"]
    & [0.5em]
    \Blin(\Cmplx)^+
    \arrow[d,rightarrow,"\btheta(\hh\pam{1})"]
    \arrow[symbol=\cong]{r}
    & [-2em]
    \Real
    \\
    \cH
    \arrow[r,rightarrow,"\bpi"']
    &
    \Blin(\cH)^+
    &
  \end{tikzcd}
\end{equation}
shows that
$\bpi(\hh) = \btheta(\hh\pam{1})\cdot\Id_\Cmplx$.
In the literature, a sum 
$\sum_n\btheta(\hh_n\pam{1})$ is called a \textit{preparation channel}.

\subsection{$\gJ$ is a natural Choi map}

\label{sec:J natural}

%

Now we come to the main point of this section, which is that
$\gJ$ is a natural unitary,
according to the naturality square
  \begin{equation}\label{eq:J natural}
\begin{tikzcd}
  \CP(\cH,\cK)
  \arrow[d,rightarrow,"\btheta(B)\blank \btheta(A^\dag)"']
  \arrow[r,leftrightarrow,"\gJ"]
  & [-1em]
  \Blin(\Blin(\cH,\cK))^+
  \arrow[d,rightarrow,"\btheta(B\blank{A^\dag})"']
      \\ [0.5em]
  \CP(\cM,\cN)
  \arrow[r,leftrightarrow,"\gJ"]
  &
  \Blin(\Blin(\cM,\cN))^+
\end{tikzcd}
\end{equation}
We need to check both commutativity, and
that the left and right edges correspond to (bi)functors.

Here is the commutativity check. Start with
$\btheta(X)$ in $\CP(\cH,\cK)$ (upper-left corner) and take it
across and down to obtain
$\btheta(B\blank A^\dag)\bpi(X)
= \bpi((B\blank A^\dag)X)
= \bpi(BX A^\dag)$.
For the down-and-across route,
note that $(\btheta(B)\blank \btheta(A^\dag))\btheta(X)
= \btheta(B)\btheta(X)\btheta(A^\dag)
= \btheta(BXA^\dag)$ by functoriality of $\btheta$.
$\gJ$ takes this to $\bpi(BXA^\dag)$.

Second, check correspondence of the left and right sides of the diagram to
functors.
It is not difficult to directly verify that the vertical arrows correspond to
functors.
However, a little effort can be saved by noting that they are compositions
of functors we have already met, namely,
\begin{equation}\nonumber
  (A,B)
  \,\mapsto\,
  B\blank A^\dag
  \,\mapsto\,
 \btheta(B\blank A^\dag),
\end{equation}
in the one case, and
\begin{equation}\nonumber
  (A,B)
  \,\mapsto\,
  (\btheta(A),\btheta(B))
  \,\mapsto\,
  \btheta(B)\blank \btheta(A)^\dag
\end{equation}
in the other.

\subsection{Uniqueness}

The argument here parallels the one we gave in \S\ref{sec:T unique}.
It all comes down to $\Id_{\Blin(\Cmplx)}$ in $\CP(\Cmplx,\Cmplx)$.
However, this time we are asking for a natural transformation which
is not only unitary, but which sends $\Id_{\Blin(\Cmplx)}$
($=\btheta(\Id_\Cmplx)$)
 to a positive operator on $\Blin(\Cmplx)$. Since
$\Blin(\Cmplx)$ is one-dimensional, there is only one choice.

\section{The natural Choi map carried on  $\cK\otimes\cH^*$}\label{sec:C}

$\gJ$ bears little resemblance to $\vJ_\Upsilon$ from the introduction.
The next step is to find something much closer, namely, the
unique natural Choi map with carrier $\cK\otimes\cH^*$.

\subsection{Definition}

The preceding two sections showed that $\gJ$ is a Choi map, and in
precisely what sense it is natural.
According to Prop.~\ref{prop:Choi}, and updating the statement with
our new notation $\btheta$,
any $(\cH,\cK)$-Choi map takes the form
\begin{equation}\nonumber
\btheta(U)\circ\gJ  
\end{equation}
for some unitary $U\colon\Blin(\cH,\cK)\rightarrow\cM$.
Furthermore, if the Choi map is to be natural in $\cH$ and $\cK$,
so should be $U$.
We have one such available, namely
\mbox{$\gT\colon\Blin(\cH,\cK)\rightarrow \cK\otimes\cH^*$}.
So, define
\begin{empheq}[box=\widefbox]{equation}\label{eq:C def}
\gC \defeq \btheta(\gT)\circ \gJ.
\end{empheq}

\subsection{Naturality and Yoneda formulas}

The naturality square for $\gT$ with $\btheta$ applied can be
pasted to the naturality square for $\gJ$, with the result
  \begin{equation}\label{eq:C natural}
\begin{tikzcd}
  (\cH,\cK)
  \arrow[r,rightarrow,"{(A,B)}"]
  & [2em]
  (\cM,\cN)
  \\ 
  \CP(\cH,\cK)
  \arrow[r,rightarrow,"\btheta(B)\blank \btheta(A)^\dag"']
  \arrow[d,leftrightarrow,"\gJ"']
  \arrow[dd,rightarrow,bend right=45,shift right=2em,"\gC"',"\cong"]
  & 
  \CP(\cM,\cN)
  \arrow[d,leftrightarrow,"\gJ"']
  \arrow[dd,rightarrow,bend left=45,shift left=2em,"\gC","\cong"']
  \\ [1em] 
  \Blin(\Blin(\cH,\cK))^+
  \arrow[r,rightarrow,"\btheta(B\blank{A^\dag})"']
  \arrow[d,rightarrow,"\btheta(\gT)"',"\cong"]
  &
  \Blin(\Blin(\cM,\cN))^+
  \arrow[d,rightarrow,"\btheta(\gT)"',"\cong"]
  \\ [1em] 
\Blin(\cK\otimes\cH^*)^+
  \arrow[r,rightarrow,"\btheta(B)\otimes \btheta(A^\ast)"']
  &
  \Blin(\cN\otimes\cM^*)^+
\end{tikzcd}
\end{equation}
For space reasons, we have rotated things here
by $90^\circ$ from our customary orientation.

\subsubsection*{Yoneda formula}

Taking $\cK=\cH$ in the naturality square (\ref{eq:C natural}) with
\begin{equation}\nonumber
  \begin{tikzcd}
    \cM\arrow[r,->,"A^\dag"]
    & \cH\arrow[r,rightarrow,"B"]
    & \cN,
  \end{tikzcd}
\end{equation}
and tracking $\Id_\cH$ along the right-then-down and down-then-right
routes of the outer square yields
\begin{equation}\label{eq:ttt}
  \begin{split}
      \gC(\btheta(B)\circ\btheta(A)^\dag)
&=\gC((\btheta(B)\blank\btheta(A^\dag))\Id_\cH)
\\ &=
(\btheta(B)\otimes \btheta({A}^\ast))\,\gC(\Id_\cH).
  \end{split}
\end{equation}
Now, each side is linear in $\btheta(B)$, so we are able to immediately
generalize to a sum $\Gamma$ of such terms, but to do the same with $A$
requires massaging $\btheta(A^\ast)$ a little, as follows:
$\btheta(A^\ast)\cdot X^\ast = (\btheta(A)\cdot X)^\ast$, or 
\begin{equation}
\btheta(A^\ast) = \ast\circ \btheta(A) \circ \ast^{-1},  
\end{equation}
with a little notational license.
Note also that
$\ast\circ \pam{A} \circ \ast^{-1} = \pam{A^*}$,
because $\inpr{A^\ast}{X^\ast}
= \pam{\inpr{A}{X}} = {\inpr{A}{X}}^\ast$.
\begin{equation}\label{eq:C Yoneda}
\gC(\Lambda\Psi^\dag)
 = (\Lambda \otimes (\ast \circ\Psi\circ \ast^{-1}))\cdot\gC(\Id).
\end{equation}
corresponds to the naturality square

\subsubsection*{}

  \begin{equation}\label{eq:C natural 2}
\begin{tikzcd}
  \CP(\cH,\cK)
  \arrow[d,rightarrow,"\Gamma \blank \Psi^\dag"']
  \arrow[r,rightarrow,"\gC","\cong"']
  & [2em]
\Blin(\cK\otimes\cH^*)^+
  \arrow[d,rightarrow,"\Gamma \otimes (\ast \circ \Psi \circ \ast^{-1})"']
  \\ [1em] 
  \CP(\cM,\cN)
  \arrow[r,rightarrow,"\gC"',"\cong"]
  &
  \Blin(\cN\otimes\cM^*)^+
\end{tikzcd}
\end{equation}
\subsection{Special forms}
%
\subsubsection{A structural formula: $\gC(\thetabox(A)) = \pibox(\gT A)$}\label{sec:C . theta}

The demonstration is
\begin{equation}\label{eq:C . theta}
\begin{split}
  \gC\circ\btheta
  &= \btheta(\gT)\circ\gJ\circ\btheta
\\  &\stackrel{(\ref{eq:J theta pi})}{=}
        \btheta(\gT)\circ\bpi
\\  &\stackrel{(\ref{eq:pi natural})}{=}
 \bpi\circ\gT.
\end{split}
\end{equation}
This relation is of the species earlier called a structural formula.
Given a Kraus decomposition $\Lambda = \sum \btheta(A_n)$, we obtain
$\gC\Lambda$ by applying the formula to each term.

\subsubsection{Action on identity: $\gC(\Id_{\Blin(\cH)}) = \pibox(I_\cH)$}\label{sec:C(Id)}

Since $\Id_{\Blin(\cH)}=\btheta(\Id_\cH)$ and $\gT(\Id_\cH) = I_\cH$,
this is immediate from (\ref{eq:C . theta}).

%
\subsubsection{A structural formula: $\gC(Z\pam{Y}) = Z\otimes Y^\ast$}
\label{sec:C(ZY-bar)}

With $Z\colon\Blin(\cK)$ and $Y\colon\Blin(\cH)$,
$\gC(Z\pam{Y})\in \Blin(\cK)\otimes\Blin(\cH)^*$.
Just examining the types strongly suggests that the formula
must be correct.
To render $Z\pam{Y}$ as a manifest composition, insert a trivial space as
\begin{equation}\nonumber
  \begin{tikzcd}
    \Blin(\cH) \arrow[r,->,"\Id_{\Cmplx}\pam{Y}"]
    & [0.5em]
    \Blin(\Cmplx) \arrow[r,->,"Z\pam{\Id}_{\Cmplx}"]
    & [0.5em]
    \Blin(\cK)
  \end{tikzcd}.
\end{equation}
This is now set to verify the formula as
\begin{align}
  \gC(Z\pam{Y})
    &= \gC((Z\pam{\Id}_\Cmplx)\circ(\Id_\Cmplx\pam{Y}))
  \nonumber
  \\  &\stackrel{(\ref{eq:C Yoneda})}{=}
        (Z\pam{\Id}_\Cmplx\otimes Y^\ast \pam{\Id_\Cmplx^\ast})
        \bpi(I_{\Cmplx})
        \nonumber
  \\  & =
        (Z\otimes Y^\ast)
        \pam{({\Id}_\Cmplx\otimes \Id_\Cmplx^\ast)}
        \bpi(I_{\Cmplx}).
  \label{eq:ZY-bar 1}
\end{align}
Now, observe that
\begin{equation}
  \begin{split}
  \bpi(I_{\Cmplx}) 
&= \bpi(1\otimes\pam{1}) 
\\ &= \bpi(1)\otimes\bpi(\pam{1}) 
\\ &= \Id_\Cmplx\otimes {\Id}_{\Cmplx^\ast}
\\ &= \Id_\Cmplx\otimes {\Id}_\Cmplx^\ast,
  \end{split}
\end{equation}
so that (\ref{eq:ZY-bar 1}) collapses to just
$Z\otimes Y^\ast$.
\subsubsection{``Inversion'': 
  $\Lambda(A)  = (\Id\otimes\pam{A^\ast}) (\gC \Lambda)$}
\label{sec:Lambda(Z)}

%
%
The formula \ref{sec:C(ZY-bar)} is very useful for proving this.
Working with the map $A\pam{\Id_\Cmplx}$ instead of $A$ will cast the
result as an instance of a Yoneda formula, even if functionally,
it appears to be an inversion.
On the one hand,
\begin{equation}\nonumber
  \begin{split}
    \Lambda\circ (A\pam{\Id}_\Cmplx)
    &= \Lambda(A)\pam{\Id}_\Cmplx
    \\  & \stackrel{\gC}{\mapsto}
          \Lambda(A)\otimes\Id_\Cmplx^\ast,
  \end{split}
\end{equation}
by the preceding special case.
On the other hand,
\begin{align}
  \Lambda\circ A\pam{\Id}_\Cmplx
 \stackrel{(\ref{eq:C Yoneda})}{=}
\Big(\Id\otimes\Id_\Cmplx^\ast\pam{A^\ast}\Big) \gC(\Lambda).
  \nonumber
  \\ = 
\Big(\Id\otimes \pam{A^\ast}\Big) \gC(\Lambda) 
\otimes\Id_\Cmplx^\ast
\end{align}

\subsection{Uniqueness}\label{sec:C unique}

$\gC$ is the unique natural $(\cH,\cK)$-Choi map with carrier
$\cK\otimes\cH^*$. The argument is only a slight variation on that
for $\gJ$. Under any natural Choi map, $\Id_{\Blin(\Cmplx)}$ must be
carried to a positive operator of norm one on the one-dimensional
space $\Cmplx\otimes\Cmplx^*$. The only choice is
$\bpi(1\otimes \pam{1})$.

\section{Conventional CJ isomorphisms}\label{sec:CV}

At this point, we have found two Choi maps
\begin{equation}\nonumber
\gJ\colon\CP(\cH,\cK)\rightarrow\Blin(\Blin(\cH,\cK))^+
\end{equation}
and
\begin{equation}\nonumber
\gC\colon\CP(\cH,\cK)\rightarrow\Blin(\cK\otimes\cH^*)^+,
\end{equation}
natural in both $\cH$ and $\cK$.
The second of these bears considerable resemblance to the
conventional Choi-Jamio{l}kowski isomorphisms in the CJ Theorem,
except that the carrier is 
$\cK\otimes\cH^*$ rather than $\cK\otimes\cH$.
Such is the price of naturality.
However, it is not the aim of this paper to replace the
ordinary Choi-Jamio{l}kowski isomorphisms, but rather to locate them in
relation to a fully natural framework. The practical-minded may think
there is nothing to be done: Merely pretend $\cH$ is $\cH^*$ via some
unitary and carry on. Well, essentially yes, but
it pays to be somewhat more deliberate about it.

\subsection{$\cK$-natural Choi maps carried on $\cK\otimes\cH$ }
\label{sec:K-natural Choi}

Recall the natural unitary $\cH\rightarrow \cH^{**}$.
It plays an important role in this section, because
$\gC$ provides a \mbox{$(\cH^*,\cK)$-Choi} map carried on
$\cK \otimes \cH^{**}$. With the aforementioned unitary, we are free
to regard this as tantamount to carrier $\cK \otimes \cH$.
But, we do not want to clutter things by writing an explicit conversion,
so we will mostly work explicitly with $\cH^{**}$, understanding that the
conversion is available on demand and can be applied without violence to
any natural construction.

With that clarification, suppose we had a $(\cH,\cK)$-Choi map
carried on $\cK\otimes\cH$ which was natural in $\cK$, but not
necessarily in $\cH$.
Then, by composition, we would have a unitary $\alpha_\cK$ from
$\CP(\cH,\cK)$ to $\CP(\cH^*,\cK)$ natural in $\cK$, with
a naturality square
\begin{equation}
    \begin{tikzcd}
  \CP(\cH,\cH)
  \arrow[r,rightarrow,"\btheta(B)\blank\btheta(\Id_\cH)"]
  \arrow[d,rightarrow,"\alpha_\cH"']
  & [2em]
  \CP(\cH,\cK)
  \arrow[d,rightarrow,"\alpha_\cK"]
  \\ 
  \CP(\cH^*,\cH)
  \arrow[r,rightarrow,"\btheta(B)\blank\btheta(\Id_\cH)"']
  & 
  \CP(\cH^*,\cK)
\end{tikzcd}
\end{equation}
(As with any of our natural transformations, we are dealing with a collection
of morphisms. However, in this case we do not yet have an explicit expression,
so the subscripts on $\alpha$ an needed to keep things straight.)
Here we can apply a Yoneda strategy, going both ways round the square
starting with $\Id_{\Blin(\cH)}$ ($=\btheta(\Id_\cH)$) in the upper-left
corner:
\begin{equation}
\alpha_\cK(\btheta(B)) = \btheta(B)\circ\alpha_\cH(\btheta(\Id_\cH)).
\end{equation}
$\alpha_\cH$ must not only be unitary as a map
\mbox{$\Blin(\Blin(\cH),\Blin(\cH)) \rightarrow\Blin(\Blin(\cH^*),\Blin(\cH))$},
but match the $\CP$ cones isometrically. Therefore,
since $\btheta(\Id_\cH)$ is extreme, $\alpha_\cH(\btheta(\Id_\cH)) = \btheta(X)$
for some $X\in\Blin(\cH^*,\cH)$. And, then,
$\alpha_\cK(\btheta(B)) = \btheta(BX)$.
For $\alpha_\cK$ to be unitary, it is necessary and sufficient that
$X$ be unitary.
Conclusion: For any unitary
$V \colon \cH \rightarrow \cH^*$,
the diagram
\begin{equation}\label{eq:K-natural}
    \begin{tikzcd}
  \CP(\cH,\cK)
  \arrow[r,rightarrow,"\btheta(B)\blank\btheta(\Id_\cH)"]
  \arrow[d,rightarrow,"\btheta(\Id_\cK)\blank\btheta(V^\dag)"']
  & [2em]
  \CP(\cH,\cN)
  \arrow[d,rightarrow,"\btheta(\Id_\cN)\blank\btheta(V^\dag)"]
  \\ [0.5em]
  \CP(\cH^*,\cK)
  \arrow[r,rightarrow,"\btheta(B)\blank\btheta(\Id_\cH)"']
  & 
  \CP(\cH^*,\cN)
\end{tikzcd}
\end{equation}
is a naturality square.

\subsection{de Pillis/Jamio{l}kowski map is not Choi}
\label{sec:de Pillis/Jamiolkowski}

The very term \textit{Choi-Jamio{l}kowski isomorphism} reflects what
seems to be a widespread assumption that Jamio{l}kowski\cite{Jamiolkowski-72}
studied a mapping which is more or less equivalent
to that suggested by Choi\cite{Choi-75}, and appearing in
(\ref{eq:usual Choi}).
In fact, it is not.
(See \cite{Jiang+Luo+Fu-13} for more about this.)
Jamio{l}kowski, following de Pillis\cite{de-Pillis-67},
studied a mapping which can be written
for $Y\in\Blin(\cK)$, $Z\in\Blin(\cH)$ as
\begin{equation}\label{eq:Jam's version}
Y\pam{Z} \stackrel{\sJ}{\mapsto} Y\otimes Z^\dag \colon \Blin(\cK\otimes\cH). 
\end{equation}
This is essentially the content of Eq. (1.4) in \cite{Jamiolkowski-72},
with a reversed order of tensor factors.
The interesting thing about this is that it is canonical in the sense
of being independent of any basis or entangled state, and it sends
monotone maps to positive operators.
However, (\ref{eq:Jam's version}) is not a Choi map. The point we wish
to make is that we are in a position to understand essentially why that
is so. Compare (\ref{eq:Jam's version}) to \S\ref{sec:C(ZY-bar)},
\begin{equation}\nonumber
Y\pam{Z} \stackrel{\gC}{\mapsto} Y\otimes Z^\ast.
\end{equation}
Thereby, we can see that $\sJ = (\Id\otimes t^{-1})\gC$.
Now, $t^{-1}$ factors through the isomorphism $\cH^{**}\cong\cH$
via $t$, so we are free to focus on
\mbox{$t\colon\Blin(\cH)\rightarrow\Blin(\cH^*)$} instead.
The transformation $t$ is natural. However, it is not $\btheta(U)$
for any unitary $U\colon\cH\rightarrow\cH^*$. Hence, according
to Prop.\ref{prop:Choi}, $\sJ$ cannot be a Choi map.

\subsection{No binatural Choi maps carried on $\cK\otimes\cH$}
\label{sec:no binatural into KxH}

Are there any Choi maps carried on $\cK\otimes\cH$ which are binatural
(i.e., natural in both $\cK$ and $\cH$)?
No, there are not. This is the reason the CJ isomorphisms need
entangled state seeds.

\S \ref{sec:K-natural Choi} found necessary and sufficient conditions for
naturality in $\cK$ alone. Naturality also in $\cH$ requires
commutativity of
\begin{equation}
    \begin{tikzcd}
  \CP(\cH,\cH)
  \arrow[r,rightarrow,"\btheta(B)\blank\btheta(A^\dag)"]
  \arrow[d,rightarrow,"\btheta(\Id_\cH)\blank\btheta(V^\dag)"']
  & [2em]
  \CP(\cH,\cK)
  \arrow[d,rightarrow,"\btheta(\Id_\cK)\blank\btheta(V^\dag)"']
  \\ 
  \CP(\cH^*,\cH)
  \arrow[r,rightarrow,"\btheta(B)\blank\btheta(\tpose{A})"']
  & 
  \CP(\cH^*,\cK)
\end{tikzcd}
\end{equation}
where $\Id_\cH$ has been replaced by arbitrary \mbox{$A\in\Blin(\cH)$}.
That implies
\begin{equation}\nonumber
\btheta(B\blank V^\dag(\tpose{A}))
=
\btheta(B\blank{A^\dag V^\dag}),
\end{equation}
or $A^\ast = VAV^\dag$, which cannot be satisfied for variable $A$.

\subsection{Choi map $\gC_V$ is natural in $\cK$}

So, among $(\cH,\cK)$-Choi maps carried on $\cK\otimes\cH$,
there are none binatural, but those natural in $\cK$ are in
one-to-one correspondence with unitaries $V\colon\cH\rightarrow\cH^*$.
Now it is time to examine them more closely.
Start with giving them names. For each such unitary, $V$, define
\begin{equation}\label{eq:C_V def}
\gC_V\Lambda \defeq \gC(\Lambda\circ\btheta(V^\dag)).
\end{equation}
corresponding to the diagram
\begin{equation}
    \begin{tikzcd}
  \CP(\cH,\cK)
  \arrow[r,rightarrow,"\btheta(B)\blank\Id"']
  \arrow[d,rightarrow,"\Id\blank\btheta(V^\dag)"',"\cong"]
  \arrow[dd,rightarrow,bend right=45,shift right=2.2em,"\gC_V"',"\cong"]
  & [1em]
  \CP(\cH,\cN)
  \arrow[d,rightarrow,"\Id\blank\btheta(V^\dag)"',"\cong"]
  \arrow[dd,rightarrow,bend left=45,shift left=2em,"\gC_V"',"\cong"]
  \\ [1em] 
  \CP(\cH^*,\cK)
  \arrow[r,rightarrow,"\btheta(B)\blank\Id"']
  \arrow[d,rightarrow,"\gC"',"\cong"]
  & 
  \CP(\cH^*,\cN)
  \arrow[d,rightarrow,"\gC"',"\cong"]
  \\ [1em] 
\Blin(\cK\otimes\cH)^+
  \arrow[r,rightarrow,"\btheta(B)\otimes \Id"']
  &
  \Blin(\cN\otimes\cH)^+
\end{tikzcd}
\end{equation}
This is a trivially commuting square for the category $\CP$ on
top, fused to a naturality square for $\gC$.
The unitary $\gC_V\colon\CP(\cH,\cK)\rightarrow \Blin(\cK\otimes\cH)$ is not
natural in $\cH$. Indeed, the idea makes no sense: $\cH$ is locked down by
$V$ so not free to vary. However, it is natural in $\cK$, as the diagram
shows.
Corresponding to this naturality, we have the Yoneda formula
\begin{empheq}[box=\widefbox]{equation}
    \gC_V \Lambda = (\Lambda\otimes\Id_{\cH^{*}}^\ast)\circ \gC_V(\Id_{\Blin(\cH)}).
 \end{empheq}
This recovers essentially (\ref{eq:usual Choi}) in the Introduction.

\subsubsection*{Action on identity}

We can rewrite $\gC_V(\Id_{\Blin(\cH)})$ in various useful ways:
\begin{align}
\gC_V(\Id_{\Blin(\cH)})
  & \stackrel{(\ref{eq:C_V def})}{=} \gC(\btheta({V^\dag}))
    \nonumber
\\  & \stackrel{(\ref{sec:C . theta})}{=}
 \bpi(\gT(V^\dag)).
      \label{eq:proj on entangled}
\end{align}
This is simply the projector onto the entangled state
determined by $V$, as in \S\ref{sec:operator/state-vector duality}.

\subsubsection*{Converting between $\gC$ and $\gC_V$}
Alternatively,
\begin{align}
\gC(\btheta({V^\dag}))
 & \stackrel{}{=} \gC(\Id_{\Blin(\cH)}\circ\, \btheta({V^\dag}))
    \nonumber
\\  & \stackrel{(\ref{eq:ttt})}{=}
      \Big(\Id\otimes\, \btheta(V^\ast)\Big) \gC(\Id_{\Blin(\cH)}).
      \label{eq:C_V(1)}
\end{align}
This shows how to convert back-and-forth directly between
$\gC\Lambda$ and $\gC_V\Lambda$.
Multiplying by $(\Lambda\otimes\Id_{\Blin(\cH)})$ and using the Yoneda formulas 
\begin{equation}\label{eq:C to C_V}
  \gC_V\Lambda = \Big(\Id_{\Blin(\cK)}\otimes\, \btheta({V^\ast})\Big)( \gC\Lambda).
\end{equation}
Equivalently,
\begin{equation}\label{eq:C_V to C}
 \gC\Lambda = \Big(\Id_{\Blin(\cK)}\otimes\, \btheta(\tpose{V})\Big)( \gC_V\Lambda).
\end{equation}

\subsection{Inversion}

\begin{equation}
  \begin{split}
    \Lambda(Z) &=
                 (\Id\otimes\pam{Z^\ast}) (\gC \Lambda)
\\ &=                 
          (\Id\otimes\, \btheta(\tpose{V})\pam{Z^\ast}) (\gC_V \Lambda)
\\ &=                 
          \Big( \Id\otimes\, \pam{\btheta(V^\dag)(Z^\ast)} \Big) (\gC_V \Lambda)
  \end{split}
\end{equation}

$\gC(\thetabox(A)) = \pibox(\gT A)$

$\gC(Z\pam{Y}) = Z\otimes Y^\ast$

\section{Trace properties}\label{sec:trace}

A Choi map carry the cone of $\CP$ maps onto the cone of positive
operators on a Hilbert space.
Taken at face value, a  ``channel/state duality'' should carry
\textit{channels}, i.e., a \textit{trace-preserving}
$\CP$ maps, onto \textit{states}, i.e., positive operators of trace 1.
How do our Choi maps do in this regard?

If $\gS$ is an $(\cH,\cK)$-Choi map, then $\Tr \circ \gS$ is a
linear functional on $\Lambda\colon\Blin(\Blin(\cH),\Blin(\cK))$.
Moreover, since $\Tr UAU^{-1} = \Tr A$, Prop.\ref{prop:Choi} implies that
for every Choi map $\gS$, this is the same linear functional.
So, let us work with $\gJ$ to see what we can find out about it.
Recall that if $A\colon\Blin(\cH)$, then 
\begin{equation}\label{eq:trace}
\Tr A = \inpr{\Id_\cH}{A}.  
\end{equation}
Applying this a couple of times, 
\begin{align}
  \Tr \gJ \Lambda
  & \stackrel{(\ref{eq:trace})}{=}
    \inpr{\Id_{\Blin(\cH,\cK)}}{\gJ\Lambda}  \nonumber
  \\  &=  \inpr{\Id_{\cK} \pam\Id_{\cH}}{\Lambda}   \nonumber
\\  & \stackrel{(\ref{eq:switch})}{=}
        \inpr{\Id_{\cK}}{\Lambda \Id_{\cH}}   \nonumber
  \\  & \stackrel{(\ref{eq:trace})}{=}
        \Tr \Lambda(\Id_{\cH})   \label{eq:Tr JLambda}.
\end{align}
Therefore, $\Tr  \gS \Lambda = \dim \cH$
in the particular case that $\Lambda$ is trace-preserving
($\CP$ or not).
$\gS \Lambda$ fails to be normalized, but only by a factor which
depends only on $\Lambda$'s domain.
This, we can live with.
Fixing it up by introducing normalization factors across the board is
not an option because thus-modified Choi maps would no longer be
natural in $\cH$.

The converse direction is hopeless. Every ray in the cone
$\Blin(\cM)^+$ intersects the affine hyperplane $\{\Tr = 1\}$, but not 
every nonzero $\CP$ map has a trace-preserving multiple.

\section{Conclusion}

The conventional CJ isomorphism(s) are too special and the general
idea of a Choi map too permissive.
With the aid of the category-theoretic concept of
\textit{natural transformation}, we can chart a fruitful middle course.
The conventional CJ isomorphisms can be put in relation to, and some of
its properties explained in terms of, other Choi maps which themselves
can be very useful. To recapitulate:
\begin{itemize}
\item $\gJ$ is the unique (bi)natural Choi map carried on
\mbox{$\Blin(\cH,\cK)$}.
\item $\gC$ is the unique binatural Choi map carried on
  \mbox{$\cK\otimes\cH^*$}.
\item There are no binatural Choi maps carried on \mbox{$\cK\otimes\cH$},
  but for each unitary $V\colon\cH\rightarrow\cH^*$, there is a
  modification $\gC_V$ of $\gC$, which is natural in $\cK$.
  These are the conventional CJ isomorphisms.
\end{itemize}
For some purposes, the unique Choi maps are a better choice than the
conventional ones. For instance,
$\gJ$ was used in \S \ref{sec:J via theta-pi} to quickly show the
existence of Kraus decompositions; this can even be the basis for
the proof of Kraus decomposition in infinite dimensions\cite{Lammert-26b}.

\acknowledgments
This work was supported by NSF MRSEC DMR-2011839.


\appendix
\section{Proof of Prop.~\ref{prop:Choi}}\label{app:Choi}

This appendix provides the proof of Prop.~\ref{prop:Choi}, which
proceeds through two Lemmas with the aid of the following concept which 
does not appear in the body of this paper.
\begin{definition*}[Self-dual cone]
  A \textbf{cone} $\mathcal{C}$ in a vector space $\cH$ is a set such that,
  whenever $\hh$ is in $\mathcal{C}$, then so every positive multiple
  $c\hh$ ($0 < c$). $\mathcal{C}$ is \textbf{self-dual} if
  \begin{equation}\nonumber
    \hh\in\mathcal{C} \;\Leftrightarrow\; \forall \hh'\in\mathcal{C},\,
    0 \le \inpr{\hh'}{\hh}. 
  \end{equation}
\end{definition*}
A prime example of a self-dual cone is the set
$\Blin(\cH)^+\subset\Blin(\cH)$ of positive operators on $\cH$.
To see that it is such, note that
(i) $0\le A\in\Blin(\cH)$ if and only if
$0 \le \inpr{\hh}{A\hh} = \inpr{\bpi(\hh)}{A}$,
and (ii) every $B\in\Blin(\cH)^+$ is a sum of operators of form (\ref{eq:pi})
(use an eigenvector decomposition).

\begin{lem}\label{lem:self-dual}
  Suppose $\mathcal{C}\subset\cH$
  and $\mathcal{C}'\subset\cK$ are self-dual cones,
  and $U\colon\cH\rightarrow\cK$
  is a unitary such that $U\mathcal{C}\subseteq \mathcal{C}'$.
  Then, $U$ maps $\extr\mathcal{C}$ bijectively onto $\extr\mathcal{C}'$
  and $\mathcal{C}$ onto $\mathcal{C}'$.
\end{lem}
\begin{proof}
  Since $U$ is unitary, i.e., preserves all inner products,
  $U\mathcal{C}$ is also a self-dual cone in $\cK$. Therefore
  \begin{equation}
    \hh\in U\mathcal{C}
   \;\Leftrightarrow\;
    \forall \hh\in\mathcal{C},\; 0\le\inpr{U\hh}{\hh}
    \; \Leftarrow \;
    \hh\in\mathcal{C}',
  \end{equation}
  with ``$\Leftarrow$'' due to the fact that $U\mathcal{C}$ is
  contained in the self-dual cone $\mathcal{C}'$,
  by hypothesis. Hence, $\mathcal{C}'=U\mathcal{C}$.

  Now, for $\hh\in\mathcal{C}$, any decomposition $\hh=\sum\hh_n$,
  with the $\hh_n$ linearly independent vectors in $\mathcal{C}$
  is carried by $U$ into a corresponding decomposition of $U\hh$.
  So, \mbox{$\hh\not\in\extr\mathcal{C}
    \Rightarrow
  U\hh\not\in\extr\mathcal{C}'$}.
The argument reverses since $U$ is unitary.
\end{proof}
\begin{lem}\label{lem:preserves positive}
  Suppose unitary map $\Gamma\colon\Blin(\cH)\rightarrow\Blin(\cK)$
maps $\Blin(\cH)^+$ onto $\Blin(\cK)^+$.
Then $\Gamma = \btheta(U)$ for some unitary $U\colon\cH\rightarrow\cK$.
\end{lem}
\begin{proof}
For any unitary $V$, $\btheta(V)$ has the indicated property of bijectively
mapping the cones of positive operators.
By considering $\btheta(V^\dag)\circ\Gamma$ instead of $\Gamma$,
we can therefore assume that $\cK=\cH$.
  
Lemma~\ref{lem:self-dual} shows that $\Gamma$ maps $\extr\Blin(\cH)^+$ into
itself. But, by the discussion around (\ref{eq:pi}),
$\extr\Blin(\cH)^+$ is $\setof{\bpi(\hh)}{\hh\in\cH}$. 
Also, unitarity of $\Gamma$ implies that
\begin{equation}
|\inpr{\hh}{\hh'}|^2
= \inpr{\bpi(\hh)}{\bpi(\hh')}
= \inpr{\Gamma\cdot\bpi(\hh)}{\Gamma\cdot\bpi(\hh')}.
\end{equation}
Thus, we have precisely the hypotheses of Wigner's famous
``unitary/anti-unitary''
theorem\cite{%
  Bargmann-64,
  Sternberg-Group_theory_physics,
  Weinberg-QFT-I,
  Cassinelli+97,
  Simon+08%
  },
which therefore implies that there is a unitary or antiunitary $U$ such that
$\Gamma\cdot\bpi(\hh) = \bpi(U\hh) = \btheta(U)\cdot\bpi(\hh)$.
It must be unitary since $\Gamma$ is linear.
Finally, $\Gamma = \btheta(U)$ follows
because $\{\bpi(\hh)\}$ spans $\Blin(\cH)$ linearly.
\end{proof}
\begin{proof}[Proof of Prop.~\ref{prop:Choi}]
Now, assume $S$ and $T$ are as in the statement of the Proposition.
It follows that $\dim \cM = \dim\cN$. Thus, there is some unitary
\mbox{$V\colon\cN\rightarrow\cM$}.
It is easy to see that $\btheta(V)$ carries $\Blin(\cN)^+$
bijectively onto $\Blin(\cM)^+$.
So, $\btheta(V)\circ S \circ T^\dag$ preserves $\Blin(\cM)^+$,
and Lemma \ref{lem:preserves positive} then implies that
\mbox{$\btheta(V)\circ S \circ T^\dag = \btheta(W)$} for a unitary
\mbox{$W\in\Blin(\cM)$}.
That is, \mbox{$S = \btheta(V^\dag W)\circ T$}.
\end{proof}

\bibliography{cj}

@article{Grabowski+Kus+Marmo-07,
Author = {Grabowski, Janusz and Kus, Marek and Marmo, Giuseppe},
Title = {On the relation between states and maps in infinite dimensions},
Journal = {Open Systems \& Information Dynamics},
Year = {2007},
Volume = {14},
Number = {4},
Pages = {355-370},
Month = {DEC},
DOI = {10.1007/s11080-007-9061-3},
ISSN = {1230-1612},
EISSN = {1793-7191},
ResearcherID-Numbers = {Kus, Marek/N-3936-2014
   Grabowski, Janusz/K-1248-2013},
ORCID-Numbers = {Kus, Marek/0000-0002-2767-3251
   Grabowski, Janusz/0000-0001-8715-2370},
Unique-ID = {WOS:000252247700002},
}

@book {Bengtsson+Zyczkowski-Book,
    AUTHOR = {Bengtsson, Ingemar and \.Zyczkowski, Karol},
     TITLE = {Geometry of quantum states},
      NOTE = {An introduction to quantum entanglement},
 PUBLISHER = {Cambridge University Press, Cambridge},
      YEAR = {2006},
      ISBN = {978-0-521-81451-5; 0-521-81451-0},
   MRCLASS = {81-01 (47N50 53C80 58B20 81P15 81P68)},
  MRNUMBER = {2230995},
MRREVIEWER = {Paul\ B.\ Slater},
       DOI = {10.1017/CBO9780511535048},
       URL = {https://doi.org/10.1017/CBO9780511535048},
}

@article{Choi-75,
Author = {Choi, MD},
Title = {Completely Positive Linear Maps on Complex Matrices},
Journal = {Linear Algebra and its Applications},
Year = {1975},
Volume = {10},
Number = {3},
Pages = {285-290},
DOI = {10.1016/0024-3795(75)90075-0},
ISSN = {0024-3795},
Unique-ID = {WOS:A1975AF31800008},
}

@article {de-Pillis-67,
    AUTHOR = {de Pillis, John},
     TITLE = {Linear transformations which preserve hermitian and positive
              semidefinite operators},
   JOURNAL = {Pacific J. Math.},
  FJOURNAL = {Pacific Journal of Mathematics},
    VOLUME = {23},
      YEAR = {1967},
     PAGES = {129--137},
      ISSN = {0030-8730,1945-5844},
   MRCLASS = {15.40 (47.00)},
  MRNUMBER = {222101},
MRREVIEWER = {H.\ Minc},
       URL = {http://projecteuclid.org/euclid.pjm/1102991990},
}

@article {Jamiolkowski-72,
    AUTHOR = {Jamio{\l}kowski, A.},
     TITLE = {Linear transformations which preserve trace and positive
              semidefiniteness of operators},
   JOURNAL = {Rep. Mathematical Phys.},
  FJOURNAL = {Reports on Mathematical Physics},
    VOLUME = {3},
      YEAR = {1972},
    NUMBER = {4},
     PAGES = {275--278},
      ISSN = {0034-4877},
       DOI = {10.1016/0034-4877(72)90011-0},
       URL = {https://doi.org/10.1016/0034-4877(72)90011-0},
}

@book {Vacchini-Foundations,
    AUTHOR = {Vacchini, Bassano},
     TITLE = {Open quantum systems---foundations and theory},
    SERIES = {Graduate Texts in Physics},
 PUBLISHER = {Springer, Cham},
      YEAR = {2024},
      ISBN = {978-3-031-58217-2; 978-3-031-58218-9},
       DOI = {10.1007/978-3-031-58218-9},
       URL = {https://doi.org/10.1007/978-3-031-58218-9},
}

@book {Benatti-23,
    AUTHOR = {Benatti, Fabio},
     TITLE = {Dynamics, information and complexity in quantum systems},
    SERIES = {Theoretical and Mathematical Physics},
   EDITION = {Second},
 PUBLISHER = {Springer, Cham},
      YEAR = {2023},
      ISBN = {978-3-031-34247-9; 978-3-031-34248-6},
   MRCLASS = {82C10 (37Axx 37Bxx 46N50 47N50 81P68 82-02 82C03)},
  MRNUMBER = {4647499},
       DOI = {10.1007/978-3-031-34248-6},
       URL = {https://doi.org/10.1007/978-3-031-34248-6},
}

@book {Wilde,
    AUTHOR = {Wilde, Mark M.},
     TITLE = {Quantum information theory},
   EDITION = {Second},
 PUBLISHER = {Cambridge University Press, Cambridge},
      YEAR = {2017},
      ISBN = {978-1-107-17616-4},
   MRCLASS = {81-01 (81-02 81P45 94A17)},
  MRNUMBER = {3645110},
       DOI = {10.1017/9781316809976},
       URL = {https://doi.org/10.1017/9781316809976},
}

@article{Jiang+Luo+Fu-13,
  title = {Channel-state duality},
  author = {Jiang, Min and Luo, Shunlong and Fu, Shuangshuang},
  journal = {Phys. Rev. A},
  volume = {87},
  issue = {2},
  pages = {022310},
  numpages = {8},
  year = {2013},
  month = {Feb},
  publisher = {American Physical Society},
  doi = {10.1103/PhysRevA.87.022310},
  url = {https://link.aps.org/doi/10.1103/PhysRevA.87.022310},
}

@article{Homa+Ortega+Koniorczyk-24,
Author = {Homa, Gabor and Ortega, Adrian and Koniorczyk, Matyas},
Title = {Choi representation of completely positive maps in brief},
Journal = {Zeitschrift f\"{u}r Naturforschung Section A-A Journal of Physical Sciences},
Year = {2024},
Volume = {79},
Number = {12},
Pages = {1123-1133},
Month = {DEC 17},
DOI = {10.1515/zna-2024-0144},
EarlyAccessDate = {OCT 2024},
ISSN = {0932-0784},
EISSN = {1865-7109},
ResearcherID-Numbers = {/AAE-6989-2021
   Koniorczyk, Mátyás/E-1373-2011},
ORCID-Numbers = {Ortega, Aian/0000-0001-5816-7274
   Koniorczyk, Mátyás/0000-0002-2710-493X},
Unique-ID = {WOS:001333930100001},
}

@article{Zyczkowski+Bengtsson-04,
Author = {Zyczkowski, K and Bengtsson, I},
Title = {On duality between quantum maps and quantum states},
Journal = {Open Systems \& Information Dynamics},
Year = {2004},
Volume = {11},
Number = {1},
Pages = {3-42},
Month = {MAR},
DOI = {10.1023/B:OPSY.0000024753.05661.c2},
ISSN = {1230-1612},
EISSN = {1793-7191},
ResearcherID-Numbers = {Zyczkowski, Karol/M-6465-2019},
ORCID-Numbers = {Bengtsson, Ingemar/0000-0002-4203-3180
   Zyczkowski, Karol/0000-0002-0653-3639},
Unique-ID = {WOS:000220488100002},
}

@article{Arrighi+Patricot-04,
	title = {On quantum operations as quantum states},
	volume = {311},
	issn = {0003-4916},
	url = {https://www.sciencedirect.com/science/article/pii/S0003491603002707},
	doi = {10.1016/j.aop.2003.11.005},
	number = {1},
	urldate = {2026-07-17},
	journal = {Annals of Physics},
	author = {Arrighi, Pablo and Patricot, Christophe},
	month = may,
	year = {2004},
	pages = {26--52},
}

@article{Frembs+Cavalcanti-24,
	title = {Variations on the {Choi}–{Jamio{\l}kowski} isomorphism},
	volume = {57},
	issn = {1751-8113, 1751-8121},
	url = {https://iopscience.iop.org/article/10.1088/1751-8121/ad5394},
	doi = {10.1088/1751-8121/ad5394},
	number = {26},
	urldate = {2026-07-14},
	journal = {Journal of Physics A: Mathematical and Theoretical},
	author = {Frembs, Markus and Cavalcanti, Eric G},
	month = jun,
	year = {2024},
	pages = {265301},
}

@misc{Preskill-course,
note={Lecture notes for Caltech Physics/Computer Science 219},
url={https://www.preskill.caltech.edu/ph229/},
}

@article{Ranade+Ali-07,
Author = {Ranade, Kedar S. and Ali, Mazhar},
Title = {The Jamiolkowski isomorphism and a simplified proof for the
   correspondence between vectors having Schmidt number k and
   k-positive maps},
Journal = {Open Systems \& Information Dynamics},
Year = {2007},
Volume = {14},
Number = {4},
Pages = {371-378},
Month = {DEC},
DOI = {10.1007/s11080-007-9062-2},
ISSN = {1230-1612},
EISSN = {1793-7191},
ORCID-Numbers = {Ali, Mazhar/0000-0002-3708-1970},
Unique-ID = {WOS:000252247700003},
}

@book {Sternberg-Group_theory_physics,
    AUTHOR = {Sternberg, S.},
     TITLE = {Group theory and physics},
 PUBLISHER = {Cambridge University Press, Cambridge},
      YEAR = {1994},
      ISBN = {0-521-24870-1},
   MRCLASS = {20-01 (20C30 20C35 22E46)},
  MRNUMBER = {1287387},
MRREVIEWER = {J.\ L\~ohmus},
}

@article{Cassinelli+97,
Author = {Cassinelli, G and De Vito, E and Lahti, PJ and Levrero, A},
Title = {Symmetry groups in quantum mechanics and the theorem of Wigner on the
   symmetry transformations},
Journal = {Reviews in Mathematical Physics},
Year = {1997},
Volume = {9},
Number = {8},
Pages = {921-941},
Month = {NOV},
DOI = {10.1142/S0129055X97000324},
ISSN = {0129-055X},
ResearcherID-Numbers = {De Vito, Ernesto/K-6354-2015},
ORCID-Numbers = {De Vito, Ernesto/0000-0002-4320-3292},
Unique-ID = {WOS:000071341100002},
}

@article{Bargmann-64,
Author = {Bargmann, V},
Title = {Note on Wigner's Theorem on Symmetry Operations},
Journal = {Journal of Mathematical Physics},
Year = {1964},
Volume = {5},
Number = {7},
Pages = {862-\&},
DOI = {10.1063/1.1704188},
ISSN = {0022-2488},
Unique-ID = {WOS:A19645659B00002},
}

@book {Weinberg-QFT-I,
    AUTHOR = {Weinberg, Steven},
     TITLE = {The quantum theory of fields. {V}ol. {I}},
 PUBLISHER = {Cambridge University Press, Cambridge},
      YEAR = {2005},
      ISBN = {0-521-55001-7; 978-0-521-67053-1; 0-521-67053-5},
   MRCLASS = {81-01 (81-02)},
  MRNUMBER = {2148466},
}

@article {Eilenberg+MacLane-45,
    AUTHOR = {Eilenberg, Samuel and MacLane, Saunders},
     TITLE = {General theory of natural equivalences},
   JOURNAL = {Trans. Amer. Math. Soc.},
  FJOURNAL = {Transactions of the American Mathematical Society},
    VOLUME = {58},
      YEAR = {1945},
     PAGES = {231--294},
      ISSN = {0002-9947,1088-6850},
   MRCLASS = {09.1X},
  MRNUMBER = {13131},
       DOI = {10.2307/1990284},
       URL = {https://doi.org/10.2307/1990284},
}

@misc{Coecke+Kissinger-16a,
      title={Categorical Quantum Mechanics I: Causal Quantum Processes}, 
      author={Bob Coecke and Aleks Kissinger},
      year={2016},
      eprint={1510.05468},
      archivePrefix={arXiv},
      primaryClass={quant-ph},
      url={https://arxiv.org/abs/1510.05468}, 
}

@article{Coecke+Kissinger-16b,
Author = {Coecke, Bob and Kissinger, Aleks},
Title = {Categorical quantum mechanics II: Classical-quantum interaction},
Journal = {International Journal of Quantum Information},
Year = {2016},
Volume = {14},
Number = {4, SI},
Month = {June},
DOI = {10.1142/S0219749916400207},
Article-Number = {1640020},
ISSN = {0219-7499},
EISSN = {1793-6918},
ORCID-Numbers = {Kissinger, Aleks/0000-0002-6090-9684},
Unique-ID = {WOS:000383806200021},
}

@article{Barnum+Duncan+Wilce-13,
Author = {Barnum, Howard and Duncan, Ross and Wilce, Alexander},
Title = {Symmetry, Compact Closure and Dagger Compactness for Categories of
   Convex Operational Models},
Journal = {Journal of Philosophical Logic},
Year = {2013},
Volume = {42},
Number = {3},
Pages = {501-523},
Month = {June},
DOI = {10.1007/s10992-013-9280-8},
ISSN = {0022-3611},
EISSN = {1573-0433},
ResearcherID-Numbers = {Duncan, Ross/M-6828-2015},
ORCID-Numbers = {Barnum, Howard/0000-0003-2336-6147
   Duncan, Ross/0000-0001-6758-1573},
Unique-ID = {WOS:000319298300005},
}

@incollection{Abramsky+Coecke-09,
Author = {Abramsky, Samson and Coecke, Bob},
Editor = {Engesser, K and Gabbay, DM and Lehmann, D},
Title = {Categorical Quantum Mechanics},
Booktitle = {Handbook of Quantum Logic and Quantum Structures: Quantum Logic},
Year = {2009},
Pages = {261-323},
DOI = {10.1016/B978-0-444-52869-8.50010-4},
ISBN = {978-0-08-093166-1; 978-0-444-52869-8},
ORCID-Numbers = {Abramsky, Samson/0000-0003-3921-6637},
Unique-ID = {WOS:000321979700009},
}

@article{Simon+08,
Author = {Simon, R. and Mukunda, N. and Chaturvedi, S. and Srinivasan, V.},
Title = {Two elementary proofs of the Wigner theorem on symmetry in quantum
   mechanics},
Journal = {Physics Letters A},
Year = {2008},
Volume = {372},
Number = {46},
Pages = {6847-6852},
Month = {NOV 17},
DOI = {10.1016/j.physleta.2008.09.052},
ISSN = {0375-9601},
EISSN = {1873-2429},
ResearcherID-Numbers = {/F-3260-2011},
Unique-ID = {WOS:000261757700007},
}

@article{Milz+Pollock+Modi-17,
author = {Milz, Simon and Pollock, Felix A. and Modi, Kavan},
title = {An Introduction to Operational Quantum Dynamics},
journal = {Open Systems \& Information Dynamics},
volume = {24},
number = {04},
pages = {1740016},
year = {2017},
doi = {10.1142/S1230161217400169},

URL ={https://doi.org/10.1142/S1230161217400169
},
eprint = {https://doi.org/10.1142/S1230161217400169},
}

@article{Pollock+18,
  title = {Non-Markovian quantum processes: Complete framework and efficient characterization},
  author = {Pollock, Felix A. and Rodr{\'{\i}}guez-Rosario, C{\'e}sar and Frauenheim, Thomas and Paternostro, Mauro and Modi, Kavan},
  journal = {Phys. Rev. A},
  volume = {97},
  issue = {1},
  pages = {012127},
  numpages = {13},
  year = {2018},
  month = {Jan},
  publisher = {American Physical Society},
  doi = {10.1103/PhysRevA.97.012127},
  url = {https://link.aps.org/doi/10.1103/PhysRevA.97.012127}
}

@article{Majewski+Tylec-13b,
  title = {Comment on ``Channel-state duality''},
  author = {Majewski, W\l{}adys\l{}aw A. and Tylec, Tomasz I.},
  journal = {Phys. Rev. A},
  volume = {88},
  issue = {2},
  pages = {026301},
  numpages = {2},
  year = {2013},
  month = {Aug},
  publisher = {American Physical Society},
  doi = {10.1103/PhysRevA.88.026301},
  url = {https://link.aps.org/doi/10.1103/PhysRevA.88.026301},
}

@misc{Lammert-26b,
      title={Existence of Kraus decomposition in infinite dimension via strongly-convergent direct process tomography}, 
      author={Paul E. Lammert},
      year={2026},
      eprint={2608.07207},
      archivePrefix={arXiv},
      primaryClass={quant-ph},
      url={https://arxiv.org/abs/2608.07207}, 
}

@book {Pierce-91,
    AUTHOR = {Pierce, Benjamin C.},
     TITLE = {Basic category theory for computer scientists},
    SERIES = {Foundations of Computing Series},
 PUBLISHER = {MIT Press, Cambridge, MA},
      YEAR = {1991},
      ISBN = {0-262-66071-7},
   MRCLASS = {18-02 (68-02)},
  MRNUMBER = {1120026},
MRREVIEWER = {Jack\ Weinstein},
}

@book {MacLane-71,
    AUTHOR = {MacLane, Saunders},
     TITLE = {Categories for the working mathematician},
    SERIES = {Graduate Texts in Mathematics},
    VOLUME = {Vol. 5},
 PUBLISHER = {Springer-Verlag, New York-Berlin},
      YEAR = {1971},
   MRCLASS = {18-02},
  MRNUMBER = {354798},
MRREVIEWER = {H.-B.\ Brinkmann},
}

@book {Awodey-06,
    AUTHOR = {Awodey, Steve},
     TITLE = {Category theory},
    SERIES = {Oxford Logic Guides},
    VOLUME = {49},
 PUBLISHER = {The Clarendon Press, Oxford University Press, New York},
      YEAR = {2006},
      ISBN = {0-19-856861-4; 978-0-19-856861-2},
   MRCLASS = {18-01 (03G30 18A25)},
  MRNUMBER = {2229319},
MRREVIEWER = {Colin\ McLarty},
       DOI = {10.1093/acprof:oso/9780198568612.001.0001},
       URL = {https://doi.org/10.1093/acprof:oso/9780198568612.001.0001},
}

@book {Riehl-16,
    AUTHOR = {Riehl, Emily},
     TITLE = {Category theory in context},
    SERIES = {Aurora Dover Modern Math Originals},
 PUBLISHER = {Dover Publications, Inc., Mineola, NY},
      YEAR = {2016},
      ISBN = {978-0-486-80903-8; 0-486-80903-X},
   MRCLASS = {18-01 (18Axx 97H99)},
  MRNUMBER = {4727501},
}
\end{document}